\documentclass{lmcs}
\usepackage[utf8]{inputenc}
\pdfoutput=1

\usepackage{lastpage}
\lmcsdoi{18}{2}{9}
\lmcsheading{}{\pageref{LastPage}}{}{}%
{Apr.~09,~2021}{May~10,~2022}{}

\usepackage[T1]{fontenc}
\usepackage[utf8]{inputenc}

\usepackage{thm-restate}
\usepackage{pdfsync}

\usepackage{extarrows}
\usepackage{microtype}
\usepackage{mathtools}
\usepackage{amsmath}
\usepackage{amsthm}
\usepackage{amssymb}
\usepackage{bbm}

\usepackage{hyperref}
\usepackage{cleveref} 

\usepackage{stmaryrd}
\newcommand{\ldbrack}{\llbracket}
\newcommand{\rdbrack}{\rrbracket}
\usepackage{tikz}
\usetikzlibrary{patterns,arrows,arrows.meta,calc}
\pgfdeclarelayer{background}
\pgfdeclarelayer{foreground}
\pgfsetlayers{background, main, foreground}
\usepackage{xspace}

\usepackage{wrapfig}
\usepackage[all]{xy}

\usepackage{scalefnt}

\newcommand{\confof}[1]{\iota(#1)}

\newcommand{\closure}[2]{\textrm{Cl}_{#1}(#2)}

\newcommand{\XX}{\mathcal X}
\newcommand{\YY}{\mathcal Y}
\newcommand{\ZZ}{\mathcal Z}

\newcommand{\pair}[2]{\tuple {#1, #2}}
\newcommand{\nowof}[1]{\text{\sc now}(#1)}

\newcommand{\orbit}[2]{\text{\sc orbit}_{#1}(#2)} 
\newcommand{\orbity}[1]{#1\text{-\sc orbits}}

\newcommand{\Conf}[1]{\semd {#1}}
\newcommand{\clockval}[1]{\mu(\X)}
\newcommand{\macval}[1]{\text{\sc Val}(#1)}
\newcommand{\assigntoval}[2]{{#2} - {#1}} 
\newcommand{\now}{t_0} 

\newcommand{\para}[1]{\paragraph{#1.}}
\newcommand{\succe}[2]{\text{\sc succ}_{#1}(#2)}

\newcommand{\A}{\mathbb A}

\newcommand{\Ackermann}{{\sc Ackermann}}
\newcommand{\HyperAckermann}{{\sc Hyper\-Acker\-mann}}

\newcommand{\Aut}[2]{\textrm{Aut}_{#1}(#2)}

\newcommand{\ignore}[1]{}

\newcommand{\wqo}{{\sc wqo}\xspace}
\newcommand{\charact}[1]{\text{char}_{#1}}

\newcommand{\powerset}[1]{{\mathcal P}(#1)}
\newcommand{\goesto}[1]{\xrightarrow{#1}}

\newcommand{\lang}[1]{L(#1)}
\newcommand{\langa}[2]{L_{#1}(#2)}
\newcommand{\langtsa}[2]{\langa{#1}{#2}} 
\newcommand{\langts}[1]{\lang{#1}}         

\newcommand{\true}{\mathbf{true}}
\newcommand{\false}{\mathbf{false}}
\newcommand{\N}{\mathbb N}
\newcommand{\Z}{\mathbb Z}

\newcommand{\Rnonnegpos}{{\mathbb R}_{\geq 0}}
\newcommand{\Rnonneg}{\Rnonnegpos}
\newcommand{\R}{{\mathbb R}}

\newcommand{\X}{\mathtt X}
\newcommand{\Y}{\mathtt Y}

\newcommand{\x}{\mathtt x}
\newcommand{\y}{\mathtt y}

\renewcommand{\L}{\ensuremath{\mathtt L}}
\newcommand{\I}{\ensuremath{\L_I}}
\newcommand{\F}{\ensuremath{\L_F}}
\newcommand{\op}{\mathtt {op}}

\newcommand{\automaton}{\tuple{\X, \Sigma, \L, \I, \F, \Delta}}

\newcommand{\extend}[3]{#3[#1 \mapsto #2]}

\newcommand{\timedwords}[1]{\mathbb{T}(#1)}
\newcommand{\timedwordsafter}[2]{\mathbb{T}_{\geq{#2}}(#1)}

\newcommand{\transition}[5]{\tuple{#1, #2, #3, #4, #5}}

\newcommand{\regions}[2]{\textsf{Reg}(#1,#2)}
\newcommand{\reg}[2]{\regions{#1}{#2}} 

\newcommand{\fract}[1]{\mathsf{fract}({#1})}

\newcommand{\DRA}{\ensuremath{\mbox {\sc dra}}\xspace}
\newcommand{\kDRA}[1]{\ensuremath{{\textsf{\sc dra}_{#1}}}\xspace}

\newcommand{\NRA}{\mbox {\sc nra}\xspace}

\newcommand{\DTA}{\textsf{\sc dta}\xspace}
\newcommand{\mDTA}[1]{\ensuremath{{\textsf{\sc dta}_{\textnormal{\_},#1}}}\xspace}
\newcommand{\kDTA}[1]{\ensuremath{{\textsf{\sc dta}_{#1}}}\xspace}
\newcommand{\kmDTA}[2]{\ensuremath{{\textsf{\sc dta}_{#1,#2}}}\xspace}
\newcommand{\kNRA}[1]{\ensuremath{{\textsf{\sc nra}_{#1}}}\xspace}

\newcommand{\kNRAg}[1]{\ensuremath{{\textsf{\sc nra}_{#1}^{\mathrm g}}}\xspace}

\newcommand{\NTA}{\textsf{\sc nta}\xspace}
\newcommand{\mNTA}[1]{\ensuremath{{\textsf{\sc nta}_{\textnormal{\_},#1}}}\xspace}
\newcommand{\kNTA}[1]{\ensuremath{{\textsf{\sc nta}_{#1}}}\xspace}
\newcommand{\kNTAe}[1]{\ensuremath{{\textsf{\sc nta}_{#1}^\varepsilon}}\xspace}
\newcommand{\kmNTA}[2]{\ensuremath{{\textsf{\sc nta}_{#1,#2}}}\xspace}

\newcommand{\sem}[1]{\ldbrack #1 \rdbrack}
\newcommand{\semd}[1]{\left\ldbrack#1\right\rdbrack}
\newcommand{\semlog}[1]{\left\ldbrack#1\right\rdbrack}
\newcommand{\prettyexists}[2]{\exists #1\cdot #2}

\newcommand{\sep}{\ | \ }
\newcommand{\tuplesmall}[1]{(#1)}
\newcommand{\tuple}[1]{\tuplesmall{#1}}
\newcommand{\set}[1]{\{ #1 \}}
\newcommand{\setof}[2]{\set{#1 \; | \; #2}}

\newcommand{\enc}[1]{\mathsf{enc}(#1)}
\newcommand{\Enc}[1]{\mathsf{Enc}(#1)}
\newcommand{\undata}[1]{\mathsf{undata}(#1)}

\newcommand{\LCM}{\textsf{LCM}\xspace}
\newcommand{\kLCM}[1]{$#1$-\textsf{LCM}\xspace}
\newcommand{\incr}[1]{#1\,\texttt{++}}
\newcommand{\decr}[1]{#1\,\texttt{-{}-}} 
\newcommand{\ztest}[1]{#1\stackrel ? {\texttt=} 0} 

\newcommand{\reachset}[1]{\text{Reach}(#1)}

\newcommand{\limplies}{\Rightarrow}

\newif\ifstartedinmathmode
\newcommand*{\st}{
  \relax\ifmmode\startedinmathmodetrue\else\startedinmathmodefalse\fi
  \ifstartedinmathmode{\;\cdot\;}\else{s.t.~}\fi%
}

\newcommand{\wrt}{w.r.t.~}

\newcommand{\card}[1]{|{#1}|}

\newcommand{\pspace}{{\sc PSpace}\xspace}

\newcommand{\nlogspace}{{\sc NLogSpace}\xspace}

\theoremstyle{remark}
\newtheorem*{claim*}{Claim}

\newcommand{\decision}[3]{\medskip\noindent {\sc #1}. \newline {\bf Input: } #2 \newline {\bf Output: } #3\medskip}

\crefname{example}{Example}{Examples}
\Crefname{example}{Example}{Examples}

\crefname{exa}{Example}{Examples}
\Crefname{exa}{Example}{Examples}

\crefname{claim}{Claim}{Claims}
\Crefname{claim}{Claim}{Claims}

\crefname{lemma}{Lemma}{Lemmas}
\Crefname{lemma}{Lemma}{Lemmas}

\crefname{lem}{Lemma}{Lemmas}
\Crefname{lem}{Lemma}{Lemmas}

\crefname{theorem}{Theorem}{Theorems}
\Crefname{theorem}{Theorem}{Theorems}

\crefname{thm}{Theorem}{Theorems}
\Crefname{thm}{Theorem}{Theorems}

\crefname{fact}{Fact}{Facts}
\Crefname{fact}{Fact}{Facts}

\crefname{cor}{Corollary}{Corollaries}
\Crefname{cor}{Corollary}{Corollaries}

\crefname{section}{Section}{Sections}
\Crefname{section}{Section}{Sections}

\newcommand{\hide}[1]{}

\newcommand{\lopen}[1]{(#1]} 
\newcommand{\ropen}[1]{[#1)} 

\crefname{thmC}{Theorem}{Theorems}
\crefname{lemC}{Lemma}{Lemmas}

\usepackage{etoolbox}
\apptocmd{\sloppy}{\hbadness 10000\relax}{}{}

\title{Determinisability of register and timed automata{\rsuper*}}
\titlecomment{{\lsuper*}This paper is an extended version of the conference paper~\cite{ClementeLasotaPiorkowski:CONCUR:2020}.}

\author[L.~Clemente]{Lorenzo Clemente~\lmcsorcid{0000-0003-0578-9103}}	
\author[S.~Lasota]{S{\l}awomir Lasota~\lmcsorcid{0000-0001-8674-4470}}	
\author[R.~Piórkowski]{\protect{\mbox{Rados\l{}aw Piórkowski}}~\lmcsorcid{0000-0002-9643-182X}}	

\address{University of Warsaw, Poland}	
\email{clementelorenzo@gmail.com, \{sl,r.piorkowski\}@mimuw.edu.pl}  
\urladdr{\url{https://sites.google.com/view/lorenzoclemente}}
\urladdr{\url{https://www.mimuw.edu.pl/~sl}, \url{https://www.mimuw.edu.pl/~rp}}

\thanks{%
    L.~Clemente was partially supported by the Polish NCN grant 2017/26/D/ST6/00201.
    S.~Lasota was partially supported by the Polish NCN grant 2019/35/B/ST6/02322 and by the ERC grant LIPA, agreement no.~683080.
    R.~Piórkowski was partially supported by the Polish NCN grant 2017/27/B/ST6/02093.
}


\begin{document}

\maketitle

\begin{abstract}
	The deterministic membership problem for timed automata asks whether
	the timed language given by a nondeterministic timed automaton
	can be recognised by a deterministic timed automaton.
	An analogous problem can be stated in the setting of register automata.
	We draw the complete decidability/complexity landscape of the deterministic membership problem,
	in the setting of both register and timed automata.
	For register automata, we prove that the deterministic membership problem is decidable
	when the input automaton is a nondeterministic one-register automaton (possibly with epsilon transitions)
	and the number of registers of the output deterministic register automaton is fixed.
	This is optimal: We show that in all the other cases the problem is undecidable, i.e.,
	when either (1) the input nondeterministic automaton has two registers or more (even without epsilon transitions),
	or (2) it uses guessing,
	or (3) the number of registers of the output deterministic automaton is not fixed.
	The landscape for timed automata follows a similar pattern.
	We show that the problem is decidable when the input automaton is a one-clock nondeterministic timed automaton
	without epsilon transitions and the number of clocks of the output deterministic timed automaton is fixed.
	Again, this is optimal: We show that the problem in all the other cases is undecidable, i.e.,
	when either (1) the input nondeterministic timed automaton has two clocks or more,
	or (2) it uses epsilon transitions,
	or (3) the number of clocks of the output deterministic automaton is not fixed.
\end{abstract}

\clearpage


\section{Introduction}

\subsection{Automata over infinite alphabets}

\para{Timed automata}

\medskip

\emph{Nondeterministic timed automata} (\NTA) are one of the most widespread models of real-time reactive systems.
They are an extension of finite automata
with real-valued clocks which can be reset and compared by inequality constraints.
The nonemptiness problem for \NTA is decidable and, in fact, \pspace-complete,
as shown by Alur and Dill in their landmark 1994 paper~\cite{AD94}.
This paved the way for the automatic verification of timed systems,
leading to mature tools such as UPPAAL~\cite{Behrmann:2006:UPPAAL4},
UPPAAL Tiga (timed games)~\cite{CassezDavidFleuryLarsenLime:CONCUR:2005},
and PRISM (probabilistic timed automata)~\cite{KwiatkowskaNormanParker:CAV:2011}.
%
The reachability problem is still a very active research area these days~\cite{FearnleyJurdziski:IC:2015,
HerbreteauSrivathsanWalukiewicz:IC:2016,
AkshayGastinKrishna:LMCS:2018,
GastinMukherjeeSrivathsan:CONCUR:2018,
GastinMukherjeeSrivathsan:CAV:2019,
GovindHerbreteauSrivathsanWalukiewicz:CONCUR:2019},
as are expressive generalisations thereof,
such as the binary reachability problem~\cite{ComonJurski:TA:1999,
Dima:Reach:TA:LICS02,
KrcalPelanek:TM:FSTTCS:2005,
FranzleQuaasShirmohammadiWorrell:IPL:2020}.
As a testimony to the model's importance,
the authors of~\cite{AD94} received the 2016 Church Award~\cite{church:award}
for the invention of timed automata.

\emph{Deterministic timed automata} (\DTA) form a strict subclass of \NTA
where the successive configuration is uniquely determined from the current one and the timed input symbol.
The class of \DTA enjoys stronger properties than \NTA,
such as decidable universality/equivalence/inclusion problems,
and closure under complementation~\cite{AD94}.
Moreover, the more restrictive nature of \DTA is needed for several applications of timed automata,
such as test generation~\cite{NielsenSkou:STTT:2003},
fault diagnosis~\cite{BouyerChevalierDSouza:FOSSACS:2005},
and learning~\cite{VerwerWeerdtWitteveen:Benelearn:2007,TapplerAichernigLarsenLorber:FOMATS:2019},
winning conditions in timed games~\cite{AsarinMaler:HSCC:1999,JurdzinskiTrivedi:ICALP:2007,BrihayeHenzingerPrabhuRaskin:ICALP:2007},
and in a notion of recognisability of timed languages~\cite{Maler:Pnueli:FOSSACS:04}.
For these reasons,
and for the more general quest of understanding the nature of the expressive power of nondeterminism in timed automata,
many researchers have focused on defining determinisable classes of timed automata,
such as strongly non-Zeno \NTA~\cite{AsarianMalerPnueliSifakis:SSSC:1998},
event-clock \NTA~\cite{AlurFixHenzinger:TCS:1999},
and \NTA with integer-resets~\cite{SumanPandyaKrishnaManasa:2008}.
The classes above are not exhaustive,
in the sense that there are \NTA recognising deterministic timed languages
not falling into any of the classes above.
%
%

Another remarkable subclass of \NTA is obtained by requiring the presence of just one clock (without epsilon transitions).
The resulting class of \kNTA 1 is incomparable with \DTA:\@
For instance, \kNTA 1 are not closed under complement (unlike \DTA),
and there are very simple \DTA languages that are not recognisable by any \kNTA 1.
Nonetheless, \kNTA 1, like \DTA, have decidable inclusion, equivalence, and universality problems~\cite{OW04,LasotaWalukiewicz:ATA:ACM08},
although the complexity is non-primitive recursive~\cite[Corollary 4.2]{LasotaWalukiewicz:ATA:ACM08}
(see also~\cite[Theorem 7.2]{OuaknineWorrel:LMCS:2007} for an analogous lower bound for the satisfiability problem of metric temporal logic).
Moreover, the nonemptiness problem for \kNTA 1 is \nlogspace-complete
(vs.~\pspace-complete for unrestricted \NTA and \DTA, already with two clocks~\cite{FearnleyJurdziski:IC:2015}),
and the binary reachability relation of \kNTA 1
can be computed as a formula of existential linear arithmetic of polynomial size,
which is not the case in general~\cite{ClementeHofmanTotzke:CONCUR:2019}.

\para{Register automata}

The theory of register automata shares many similarities with that of timed automata.
\emph{Nondeterministic register automata} (\NRA) have been introduced by Kaminski and Francez around the same time as timed automata. They were defined as an extension of finite automata with finitely many registers
that can store input values and be compared with equality and disequality constraints.
The authors have shown, amongst other things, that nonemptiness is decidable~\cite[Theorem 1]{FK94}.
It was later realised that the problem is, in fact, \pspace-complete~\cite[Theorems 4.3 and Theorem 5.1]{DL09}.
The class of \NRA recognisable languages is not closed under complementation~\cite[Proposition 5]{FK94}; moreover, universality (and thus equivalence and inclusion) of \NRA is undecidable~\cite[Theorem 5.1]{NevenSchwentickVianu:TOCL:2004}
(already for \NRA with two registers~\cite[Theorem 5.4]{DL09}).

One way to regain decidability is to consider \emph{deterministic register automata} (\DRA),
which are effectively closed under complement and thus have a decidable inclusion (and thus universality and equivalence) problem%
\footnote{In fact, even the inclusion problem $\lang A \subseteq \lang B$ with $A$ an \NRA and $B$ a \DRA is decidable.}.
\DRA also provide the foundations of learning algorithms for data languages~\cite{MoermanSammartinoSilvaKlinSzynwelski:POPL2017}.
A recent result completing the theory of register automata has shown that a data language is \DRA recognisable if, and only if,
both this language and its complement are \NRA recognisable~\cite{KlinLasotaTorunczyk:FOSSACS:2021}.

As in the case of timed automata, it has been observed that restricting the number of registers results in algorithmic gains.
Already in the seminal work of Kaminski and Francez, it has been proved that the inclusion problem $\lang A \subseteq \lang B$
is decidable when $A$ is an \NRA and $B$ is an \NRA with one register~\cite[Appendix A]{FK94},
albeit the complexity is non-primitive recursive in this case~\cite[Theorem 5.2]{DL09}.

\subsection{The deterministic membership problem}

\para{Timed automata}

\medskip

The \emph{\DTA membership problem} asks, given an \NTA,
whether there exists a \DTA recognising the same language.
There are two natural variants of this problem,
which are obtained by restricting the resources available to the sought \DTA\@.
Let $k \in \N$ be a bound on the number of clocks,
and let $m \in \N$ be a bound on the maximal absolute value of numerical constants.
The \emph{\kDTA k}
and \emph{\kmDTA k m membership problems}
are the restriction of the problem above
where the \DTA is required to have at most $k$ clocks,
resp., at most $k$ clocks and the absolute value of maximal constant bounded by $m$.
Notice that we do not bound the number of control locations of the \DTA,
which makes the problem non-trivial.
(Indeed, there are finitely many \DTA with a bounded number of clocks, control locations, and maximal constant.)

Since untimed regular languages are deterministic,
the \kDTA k membership problem can be seen as a quantitative generalisation of the regularity problem.
For instance, the $\kDTA 0$ membership problem is precisely the regularity problem
since a timed automaton with no clocks is the same as a finite automaton.
We remark that the regularity problem is usually undecidable
for nondeterministic models of computation generalising finite automata,
e.g., context-free grammars/pushdown automata~\cite[Theorem 6.6.6]{Shallit:Book:2008},
labelled Petri nets under reachability semantics~\cite{ValkVidal-Naquet:Petri:Regular:1981}, Parikh automata~\cite{CadilhacFinkelMcKenzie:NCMA:2011}, etc.
One way to obtain decidability is to either restrict the input model to be deterministic
(e.g.,~\cite{Valiant:Regularity:DPDA:JACM:1975,ValkVidal-Naquet:Petri:Regular:1981,BaranyLodingSerre:STACS:2006}),
or to consider more refined notions of equivalence,
such as bisimulation (e.g.,~\cite{ParysGoller:LICS:2020}).

This negative situation is generally confirmed for timed automata.
For every number of clocks $k\in\N$ and maximal constant $m$,
the \DTA, \kDTA k, and \kmDTA k m membership problems are known to be undecidable
when the input \NTA has $\geq 2$ clocks,
and for 1-clock \NTA with epsilon transitions~\cite{Finkel:FORMATS:2006,Tripakis:IPL:2006}.
To the best of our knowledge,
the deterministic membership problem was not studied before when the input automaton is \kNTA 1 without epsilon transitions.

\para{Register automata}

The situation for register automata is similar to, and simpler than, timed automata.
The \emph{\kDRA k membership problem} asks, given an \NRA, whether there exists a \DRA with $k$ registers recognising the same language,
and the \emph{\DRA membership problem} is the same problem with no apriori bound on the number of registers of the deterministic acceptor.
Deterministic membership problems for register automata do not seem to have been considered before in the literature.

\subsection{Contributions}

We complete the study of the decidability border for the deterministic membership problem
initiated for timed automata in~\cite{Finkel:FORMATS:2006,Tripakis:IPL:2006},
and we extend these results to register automata.

\para{Upper bounds}

Our main result is the following.
\begin{restatable}{thm}{thmkDTAmemb}%
    \label{thm:kDTA:memb}
    The \kDTA{k} and \kmDTA k m membership problems
    are decidable for \kNTA{1} languages.
\end{restatable}
Our decidability result contrasts starkly
with the abundance of undecidability results for the regularity problem.
We establish decidability by showing that if a \kmNTA 1 m recognises a \kDTA k language,
then, in fact, it recognises a \kmDTA k m language
and, moreover, there is a computable bound on the number of control locations of the deterministic acceptor
 (c.f.~Lemma~\ref{thm:k-DTA-char}).
This provides a decision procedure
since there are finitely many different \DTA once the number of clocks, the maximal constant, and the number of control locations are fixed.


In our technical analysis, we find it convenient to introduce the so-called
\emph{always resetting} subclass of \kNTA k. 
These automata are required to reset at least one clock at every transition
and are thus of expressive power intermediate between \kNTA {k-1} and \kNTA k.
Always resetting \kNTA 2 are strictly more expressive than \kNTA 1: For instance, the language of timed words of the form
$(a, t_0) (a, t_1) (a, t_2)$ \st $t_2 - t_0 > 2$ and $t_2 - t_1 < 1$ can be recognised by an always resetting \kNTA 2
but by no \kNTA 1.
Despite their increased expressive power,
always resetting \kNTA 2
still have a decidable universality problem (the well-quasi order approach of~\cite{OW04} goes through),
which is not the case for \kNTA 2.
Thanks to this restricted form,
we are able to provide in Lemma~\ref{thm:k-DTA-char} an elegant characterisation of those \kNTA 1 languages
which are recognised by an always resetting \kDTA k.

We prove a result analogous to \cref{thm:kDTA:memb} in the setting of register automata.
\begin{restatable}{thm}{thmDRAmemb}%
    \label{thm:kDRA:memb}
	The \kDRA{k} membership	 problem is decidable for \kNRA{1} languages.
\end{restatable}
Thanks to the effective elimination of $\varepsilon$-transition rules from \kNRA 1 (c.f.~Lemma~\ref{lem:NRA:epsilon-transitions}),
the decidability result above also holds for data languages presented as \kNRA 1 with $\varepsilon$-transition rules.

\para{Lower bounds}

We complement the decidability results above by showing that the deterministic membership problem becomes undecidable
if we do not restrict the number of clocks/registers of the deterministic acceptor.
\begin{restatable}{thm}{thmUndec}%
    \label{thm:undecidability}
    The \DTA and \mDTA m ($m > 0$) membership problems are undecidable for \kNTA 1 without epsilon transitions.
\end{restatable}
\noindent
\cref{thm:undecidability} improves on an analogous result from~\cite[Theorem 1]{Finkel:FORMATS:2006} for \kNTA 2.
We obtain a similar undecidability result in the setting of register automata:
\begin{restatable}{thm}{thmUndecReg}%
    \label{thm:undecidabilityreg}
    The \DRA membership problem is undecidable for \kNRA 1.
\end{restatable}

The following lower bounds further refine the analysis from~\cite{Finkel:FORMATS:2006}
in the case of a fixed number of clocks of a deterministic acceptor.
%

\begin{restatable}{thm}{thmEasyUndecidability}%
	\label{thm:easy-undecidability}
	For every fixed $k, m\in \N$, the \kDTA k and \kmDTA k m membership problems are:
	\begin{itemize}
	\item undecidable for \kNTA 2,
	\item undecidable for \kNTAe 1 (with epsilon transitions),
	\item \HyperAckermann-hard for \kNTA 1.
	\end{itemize}
\end{restatable}

\noindent
A similar landscape holds for register automata,
where the deterministic membership problem for a fixed number of registers of the deterministic acceptor
remains undecidable when given in input either a \kNRA 2 or a \kNRA 1 with guessing.
(Register automata with guessing are a more expressive family of automata
where a register can be updated with a data value not necessarily coming from the input word, i.e., it can be \emph{guessed}.
We omit a formal definition since we will not need to explicitly manipulate such automata in this paper.)
In the decidable case of a \kNRA 1 input, the problem is nonetheless not primitive recursive.

\begin{restatable}{thm}{thmDRAlowerbounds}%
	\label{thm:reg-lower-bound}
	Fix a $k \geq 0$.
	The \kDRA k membership problem is:
	\begin{enumerate}
		\item undecidable for \kNRA 2,
		\item undecidable for \kNRAg 1 (\kNRA 1 with guessing), and
		\item not primitive recursive (\Ackermann-hard) for \kNRA 1.
	\end{enumerate}
\end{restatable}

\para{Related research}

Many works have addressed the construction of a \DTA equivalent to a given \NTA
(see~\cite{BertrandStainerJeronKrichen:FMSD:2015} and references therein).
However, since the general problem is undecidable,
one has to either sacrifice termination,
or consider deterministic under/over-approximations.
In a related line of work,
we have shown that the \emph{deterministic separability problem} is decidable for the full class of \NTA, when the number of clocks of the separator is given in the input~\cite{ClementeLasotaPiorkowski:ICALP:2020}.
This contrasts with the undecidability of the corresponding membership problem.
The deterministic separability problem for register automata has not been studied in the literature.
Decidability of the deterministic separability problem when the number of clocks/registers of the separator is not provided remains a challenging open problem.


\section{Automorphisms, orbits, and invariance}\label{sec:atoms}

This section contains preliminary definitions needed both for register and timed automata.

\para{Atoms}
Let $\A$ be a structure (whose elements are called \emph{atoms})
providing a data domain on which register and timed automata operate.
In the case of register automata, we will primarily be concerned with \emph{equality atoms} $\tuple{\A, =}$,
where $\A$ is a countable set, and the signature contains the equality symbol only.
However---as we discuss at the end of Section~\ref{sec:upperbound}---our positive results generalise to
other atoms, for instance to densely ordered atoms $\tuple{\R, \leq}$,
where $\R$ is the set of real numbers with the natural order ``$\leq$''.
In the case of timed automata, we will consider \emph{timed atoms} $\tuple{\R, \leq, {+1}}$, which
extend densely ordered atoms with the increment function ``${+1}$''.

\para{Automorphisms}
Let $S \subseteq \A$ be a (possibly empty) finite set of atoms.
An {\emph{$S$-automorphism}} is a bijection $\pi : \A \to \A$ that is the identity on $S$,
i.e., $\pi(a) = a$ for every $a \in S$,
and which preserves the atom structure.
The latter condition is trivially satisfied for equality atoms.
In the case of densely ordered atoms, it also demands monotonicity: $a \leq b$ implies $\pi(a) \leq \pi(b)$.
For instance, a non-trivial automorphism of densely ordered atoms is $\pi(a) = 2 \cdot a$,
which is also a $\set 0$-automorphism since $\pi(0) = 0$.
In the case of timed atoms automorphisms need to additionally preserve integer differences:
$\pi(a+1) = \pi(a)+1$, for every $a, b \in \R$.
For instance, if $\pi(3.4)=4.5$, then the last condition necessarily implies $\pi(5.4) = 6.5$ and $\pi(-3.6) = -2.5$.
When $S = \emptyset$, we just say that $\pi$ is an automorphism.
Let $\Aut S \A$ denote the set of all $S$-automorphisms of $\A$,
and let $\Aut {} \A = \Aut \emptyset \A$.

Informally speaking, by \emph{sets with atoms} we mean sets whose elements are either sets or atoms
(for a formal definition of the hierarchy thereof, we refer to~\cite{atombook}).
If $X$ is a set with atoms, then $\pi(X)$ is the set with atoms
which is obtained by replacing every atom $a \in \A$ occurring in $X$ by $\pi(a) \in \A$.
We present some commonly occurring concrete examples, relying on standard set-theoretic encodings of tuples, functions, etc.
Let $\Sigma$ be a finite alphabet.
A \emph{data word} is a finite sequence
\begin{align}%
  \label{eq:data-word}
  w = \tuple{\sigma_0, a_0} \cdots \tuple{\sigma_n, a_n} \in (\Sigma \times \A)^*
\end{align}
of pairs $\tuple{\sigma_i, a_i}$ consisting of an input symbol $\sigma_i \in \Sigma$
and an atom $a_i \in \A$.
An automorphism $\pi$ acts on a data word $w$ as above point-wise:
$\pi(w) = \tuple{\sigma_0, \pi(a_0)} \cdots \tuple{\sigma_n, \pi(a_n)}$.

Let $\X$ be a finite set of register names and let $\A_\bot = \A \cup \set \bot$,
where $\bot \not\in \A$ represents an undefined value.
A \emph{register valuation} is a mapping 
$\mu : \A_\bot^\X$
assigning an atom (or $\bot$) $\mu(x)$ to every register $x \in \X$.
An automorphism $\pi$ acts on a register valuation $\mu$ as $\pi(\mu)(x) = \pi(\mu(x))$ for every $x \in \X$,
i.e., $\pi(\mu) = \pi \circ \mu$,
where we assume that $\pi(\bot) = \bot$.

More generally, if $X$ is a set with atoms, then $\pi$ acts on $X$ pointwise as $\pi(X) = \setof{\pi(x)}{x \in X}$.

\para{Orbits and invariance}
A set with atoms
$X$ is \emph{$S$-invariant} if $\pi(X) = X$ for every {$S$-automorphism} $\pi \in \Aut S \A$.
%
Notice that $\pi$ does not need to be the identity on $X$ for $X$ to be $S$-invariant.
%
A set $X$ is \emph{invariant}%
\footnote{The term \emph{equivariant} is also often used in the literature instead of invariant.
Also, in the case of an $S$-invariant set $X$,
one can also find the term \emph{$S$-supported},
and call $S$ a \emph{support} of $X$; see e.g.~\cite{atombook,LMCS14}.}
if it is $S$-invariant with $S = \emptyset$.
The \emph{$S$-orbit} of an element $x \in X$
(which can be an arbitrary object on which the action of  automorphisms is defined)
is the set \[\orbit S x = \setof{\pi(x) \in X}{\pi \in \Aut S \A}\]
of all elements $\pi(x)$ which can be obtained by applying some $S$-automorphism $\pi$ to $x$.
The \emph{orbit} of $x$ is just its $S$-orbit with $S = \emptyset$, written $\orbit {} x$.
Clearly $x, y \in X$ have the same $S$-orbit
$\orbit S x = \orbit S y$
if, and only if, $\pi(x) = y$ for some $\pi \in \Aut S \A$.

The \emph{$S$-orbit closure} of a set $X$ (or just \emph{$S$-closure})
is the union of the $S$-orbits of its elements:
\[\closure S X = \bigcup_{x \in X} \orbit S x.\]
In particular, the $S$-orbit of $x$ is the $S$-closure of the singleton set $\set{x}$:
$\orbit S x = \closure S {\set{x}}$.
The $\emptyset$-closure $\closure {} X$ we briefly call the \emph{closure} of $X$.
The following fact characterises invariance in term of closures.
\begin{fact}\label{fact:inv:clos}
  A set $X$ is $S$-invariant if, and only if, $\closure S X = X$.
\end{fact}
\begin{proof}
  The ``only if'' direction follows the definition of $S$-invariance.
  For the ``if'' direction, observe that
  $\closure S X = X$ implies $\pi(X)\subseteq X$.
  The opposite inclusion stems from $S$-automorphisms' closure under inverse:
  $\pi^{-1}(X) \subseteq X$, hence $X\subseteq \pi(X)$.
\end{proof}

\section{Register automata}

In this section, we define register automata over equality atoms $\tuple{\A, =}$.
However, all the definitions below can naturally be generalised to any atoms satisfying some mild assumptions,
as explained later in \cref{sec:OtherAtoms}.

\para{Constraints}
A \emph{constraint} is a quantifier-free formula $\varphi$ generated by the grammar
\begin{align}\label{eq:constr}
    \varphi, \psi \ ::\equiv\ \true \sep \false \sep \x = \y \sep \x = \bot \sep \neg \varphi \sep \varphi \land \psi \sep \varphi \lor \psi,
\end{align}
where $\x, \y$ are variables and $\bot$ is a special constant denoting an undefined data value.
A \emph{valuation} is a function $\mu \in \A_\bot^\X$ assigning a data value (or $\bot$) $\mu(\x)$
to every variable $\x \in \X$.
The satisfaction relation $\mu \models \varphi$ holds whenever the valuation $\mu$ satisfies the formula $\varphi$,
and it is defined in the standard way.
The \emph{semantics} of a constraint $\varphi(\x_1, \dots, \x_n)$ with $n$ free variables $\X = \set{\x_1, \dots, \x_n}$
is the set of valuations satisfying it:
$\sem\varphi=\setof{\mu \in \A_\bot^\X}{\mu \models \varphi}$.
Using~\cite[Lemma 7.5]{atombook} we easily deduce:
\begin{clm}\label{claim:def-inv}
  Subsets of $\A_\bot^\X$ definable by constraints are exactly the invariant subsets of $\A_\bot^\X$.
\end{clm}
%
%

\para{Register automata}

Let $\Sigma$ be a finite alphabet, and let $\Sigma_\varepsilon = \Sigma \cup \set \varepsilon$
be $\Sigma$ with the addition of the empty word $\varepsilon$.
A (nondeterministic) \emph{register automaton} (\NRA) is a tuple $A = \automaton$
where $\X = \set{\x_1, \dots, \x_k}$ is a finite set of register names,
$\Sigma$ is a finite alphabet,
$\L$ is a finite set of \emph{control locations},
of which we distinguish those which are \emph{initial} $\L_I \subseteq \L$, resp.,
\emph{final} $\L_F \subseteq \L$,
and $\Delta$ is a set of \emph{transition rules}.
We have two kinds of transition rules.
A \emph{non-$\varepsilon$-transition rule} is of the form
\begin{align}%
  \label{eq:reg:trans-rule}
  	\transition p \sigma \varphi \Y q \in \Delta
\end{align}
and it means that from control location $p \in \L$ we can go to $q \in \L$
by reading input symbol $\sigma \in \Sigma$,
provided that the transition constraint $\varphi(\x_1, \dots, \x_k, \y)$
holds between the current registers $\x_1, \dots, \x_k$ and the input data value $\y$ being currently read;
finally, all registers in $\Y \subseteq \X$ are updated to store the input data value $\y$.
An \emph{$\varepsilon$-transition rule} is of the form
$\tuple {p, \varepsilon, \varphi, q} \in \Delta$
and it means that from control location $p \in \L$ we can go to $q \in \L$,
but no input is read,
provided that the transition constraint $\varphi(\x_1, \dots, \x_k)$
holds between the current registers $\x_1, \dots, \x_k$.

Formally, the semantics of a register automaton $A$ as above is provided in terms of an infinite transition system
$\sem A = \tuple{C, C_I, C_F, \to}$,
where $C = \L \times \A_\bot^\X$ is the set of \emph{configurations},
which are pairs $\tuple{p, \mu}$ consisting of a control location $p \in \L$
and a register valuation $\mu \in \A_\bot^\X$.
Amongst them, $C_I = \L_I \times \set{\lambda \x \cdot \bot} \subseteq C$
is the set of \emph{initial configurations},
i.e., configurations of the form $\tuple{p, \mu}$ with $p \in \L_I$
and $\mu(\x) = \bot$ for all registers $\x$,
and $C_F = \L_F \times \A_\bot^\X$ is the set of \emph{final configurations},
i.e., configurations of the form $\tuple{p, \mu}$ with $p \in \L_F$
(without any further restriction on $\mu$).
The set of transitions ``$\to$'' is determined as follows.
For a valuation $\mu \in \A_\bot^\X$, a set of registers $\Y \subseteq \X$, and an atom%
\footnote{It suffices to consider non-$\bot$ values $a \neq \bot$
since we never need to reset a register to the undefined value $\bot$.}
$a \in \A$,
let $\extend \Y a \mu$
be the valuation which is $a$ on $\Y$ and agrees with $\mu$ on $\X \setminus \Y$.
Every non-$\varepsilon$-transition rule~\eqref{eq:reg:trans-rule}
induces a transition between configurations
\begin{align*}
  \tuple {p, \mu} \goesto {\sigma, a} \tuple {q, \extend \Y a \mu}
\end{align*}
labelled by $\tuple{\sigma, a} \in \Sigma \times \A$,
provided that the current valuation $\mu$ satisfies the constraint $\varphi$
when variable $\y$ holds the input atom $a$,
i.e., $\extend \y a \mu \models \varphi(\x_1, \dots, \x_k, \y)$.
%
%
Similarly, an $\varepsilon$-transition rule $\tuple{p, \varepsilon, \varphi, q} \in \Delta$
induces an $\varepsilon$-labelled transition $\tuple {p, \mu} \goesto \varepsilon \tuple {q, \mu}$
whenever $\mu \models \varphi(\x_1, \dots, \x_k)$.
Finally, in order to deal with $\varepsilon$-transitions,
we stipulate that whenever we have transitions
$\tuple{p, \mu} \goesto b \_ \goesto \varepsilon \tuple{q, \nu}$,
$\tuple{p, \mu} \goesto \varepsilon \_ \goesto b \tuple{q, \nu}$, or
$\tuple{p, \mu} \goesto \varepsilon \_ \goesto b \_ \goesto \varepsilon \tuple{q, \nu}$
with $b \in (\Sigma \times \A) \cup \set \varepsilon$,
then we also have the transition $\tuple{p, \mu} \goesto b \tuple{q, \nu}$.
A \emph{run} of $A$ \emph{over} a data word $w$ as in~\eqref{eq:data-word}
\emph{starting} in configuration $\tuple {p, \mu}$
and \emph{ending} in configuration $\tuple {q, \nu}$
is a labelled path $\rho$ in $\sem A$ 
of the form
\begin{align}%
  \label{eq:run}
  \rho = \tuple {p, \mu} \goesto {\sigma_0, a_0} \tuple {p_0, \mu_0} \goesto {\sigma_1, a_1} \cdots \goesto {\sigma_n, a_n} \tuple{q, \nu}.
\end{align}
The run $\rho$ is accepting if its ending configuration is accepting.
The language \emph{recognised} by configuration $\tuple{p, \mu}$ is the set of data words labelling accepting runs:
\begin{align*}
	\lang {p, \mu} = \setof{w\in (\Sigma \times \A)^*}{\sem A \text{ has an accepting run over } w \text{ starting in } \tuple{p, \mu}}.
\end{align*}
The language recognised by the automaton $A$ is the union of the languages recognised by its initial configurations,
$\lang A = \bigcup_{c \in C_I} \lang c$.
A configuration is \emph{reachable} if it is the ending configuration in a run starting in an initial configuration.

\begin{rem}
Register automata, as defined above, are \emph{without guessing},~i.e.\ an automaton can only store in a register
an atom appearing in the input.
\end{rem}

It turns out that \NRA with $\varepsilon$-transition rules are as expressive as \NRA without $\varepsilon$-transition rules.

\begin{lemC}[Excercise 2 in~\cite{atombook}]%
  \label{lem:NRA:epsilon-transitions}
  For every $k\in\N$ and \kNRA{k} one can effectively build an \kNRA{k} without $\varepsilon$-transition rules recognising the same language.
\end{lemC}

%

For this reason, from this point on, we deal exclusively with \NRA without $\varepsilon$-transition rules,
and we tacitly assume that an \NRA does not contain $\varepsilon$-transition rules.

\para{Deterministic register automata}
A register automaton $A$ is \emph{deterministic} (\DRA) if it has precisely one initial location $\I = \set{p_I}$
and, for every two rules
$\transition{p}{\sigma}{\varphi}{\Y}{q}$ and $\transition{p}{\sigma}{\varphi'}{\Y'}{q'}$
starting in the same location $p$, over the same input symbol $\sigma$
and with jointly satisfiable guards $\semlog{\varphi \land \varphi'} \neq \emptyset$,
we necessarily have $\Y = \Y'$ and $q = q'$.
Hence $A$ has at most one run over every data word $w$.
A \DRA can be easily transformed into a \emph{total} one,
i.e., one where for every location $p\in \L$ and input symbol $\sigma\in\Sigma$,
the sets defined by the constraints
$\setof{\semlog \varphi}{\prettyexists{\Y, q}{p \goesto{\sigma, \varphi, \Y} q}}$
are a partition of all valuations $\A_\bot^{\X \cup \set \y}$.
Thus, a total \DRA has exactly one run over every timed word $w$.

There are other syntactical constraints, more or less restrictive,
which would guarantee uniqueness of runs.
Our choice is motivated by the fact that every \NRA which has unique runs
can be transformed into an equivalent \DRA,
albeit this may cause an exponential increase in the number of control locations.
Nonetheless, since we deal with computational problems which are either undecidable or non-primitive recursive,
this shall not be a concern in the rest of the paper.

We write \kNRA k, resp. \kDRA k, for the class of $k$-register \NRA, resp. \DRA,
and we say that a data language is an \NRA language, \DRA language, \kDRA k language, etc.,
if it is recognised by a register automaton of the respective type.

\begin{exa}\label{example:L1reg}
    Let $\Sigma = \set \sigma$ be a unary alphabet.
	As an example of a language $L$ recognised by an \kNRA{1},
	but not by any \DRA,
    consider the set of data words where the last atom reappears earlier, i.e., words of the form:
		$(\sigma, a_1) \cdots (\sigma, a_n)$
	where $a_i = a_n$ for some $1\leq i < n$.
	The language $L$ is recognised by the \kNRA{1} $A = \automaton$
	with one register $\X = \set \x$ and three locations $\L = \set{p, q, r}$,
	of which $\I = \set p$ is initial and $\F = \set r$ is final, and transition rules
	\begin{align*}
		& \transition{p}{\sigma}{\true}{\emptyset}{p} \qquad
		\transition{p}{\sigma}{\true}{\set{\x}}{q} \qquad
		\transition{q}{\sigma}{\x\neq\y}{\emptyset}{q}  \qquad
		\transition{q}{\sigma}{\x=\y}{\emptyset}{r}.
	\end{align*}
	Intuitively, the automaton waits in $p$ until it guesses that the next input $a_i$ will be appearing at the end of the word as well,
	at which point it moves to $q$ by storing $a_i$ in the register.
	From $q$, the automaton can accept by going to $r$ exactly when the atom stored in the register reappears in the input.
	The language $L$ is not recognised by any \DRA
	since, intuitively, any deterministic acceptor needs to store unboundedly many different atoms $a_i$.
\end{exa}

\para{One-register automata}

Nondeterministic register automata with just one register enjoy stronger algorithmic properties
than the full class of nondeterministic register automata.
It was already observed in Kaminski and Francez's seminal paper that the inclusion problem becomes decidable%
\footnote{A \emph{window} in the terminology of~\cite{FK94} corresponds to a register in this paper's terminology.
\cite[Appendix A]{FK94} shows that the inclusion problem $\lang A \subseteq \lang B$ is decidable when $B$ is a two-window automaton. 
Due to the semantics of window reassignment of~\cite{FK94},
two-window automata are of intermediate expressive power between one-register automata and two-register automata.}.%
\begin{thm}[\protect{c.f.~\cite[Appendix A]{FK94}}]%
  \label{thm:FK94}
	For $A \in \NRA$ and $B \in \kNRA 1$ the language inclusion problem $\lang A \subseteq \lang B$ is decidable.
\end{thm}

We immediately obtain the following corollary, which we will use in \cref{sec:upperboundreg}.
\begin{cor}\label{lem:OW04reg}
	For $A \in \DRA$ and $B \in \kNRA 1$ the language equality problem $\lang A = \lang B$ is decidable.
\end{cor}

\begin{proof}
  The inclusion $\lang A \subseteq \lang B$ can be checked as a special instance of \cref{thm:FK94}.
  The reverse inclusion $\lang B \subseteq \lang A$ reduces to checking emptiness of a product construction of $B$ with the complement of $A$.
\end{proof}

\para{Invariance of register automata}

The following lemma expresses the fundamental invariance properties of register automata.
Given a valuation $\mu$ of registers $\X$, by $\mu(\X)\subseteq \A$ we mean the set of atoms stored
in registers: $\mu(\X) = \setof{\mu(\x)}{\x\in \X, \ \mu(\x)\in\A}$.
Automorphisms act on configurations by preserving the control location:
$\pi\tuple{p, \mu} = \tuple{p, \pi(\mu)} = \tuple{p, \pi \circ \mu}$.
\begin{restatable}[Invariance of \NRA]{lem}{NRAInvariance}%
  \label{lemma:NRA:invariance}
  \begin{enumerate}
    \item The transition system $\sem A$ is invariant:
    If $c \goesto {\sigma, a} d$ in $\sem A$ and $\pi$ is an automorphism,
    then $\pi(c) \goesto {\sigma, \pi(a)} \pi(d)$ in $\sem A$.

    \item The function $\lang{\_}$ mapping a configuration $c$ to the language $\lang c$ it recognises from $c$ is invariant:
    For all automorphisms $\pi$, $\lang {\pi(c)} = \pi(\lang c)$.

    \item The language $\lang {p, \mu}$ recognised from a configuration $\tuple{p, \mu}$ is $\mu(\X)$-invariant:
    For all $\mu(\X)$-automorphims $\pi$,
    $\pi(\lang {p, \mu}) = \lang {p, \mu}$.
  \end{enumerate}
\end{restatable}
\noindent
We refrain from proving Lemma~\ref{lemma:NRA:invariance}, since proofs of analogous invariance properties,
in the more involved setting of timed automata, are provided later (Facts~\ref{fact:equivariant:trans}--\ref{fact:invariantbase}
in Section~\ref{sec:inv}).
For the  proof of (1) in the setting of equality atoms, we refer the reader to~\cite[Sect.1.1]{atombook}; the other points are
readily derivable from (1).

\para{Deterministic membership problems}

Let $\mathcal X$ be a subclass of \NRA\@.
We are interested in the following family of decision problems:

\decision{
$\mathcal X$ membership problem}
{A register automaton $A\in$ \NRA\@.}
{Does there exist a
$B\in \mathcal X$ \st $\lang A = \lang B$?}

We study the decidability status of the $\mathcal X$ membership problem
where $\mathcal X$ ranges over \DRA and \kDRA k (for every fixed number of registers $k$).
Example~\ref{example:L1reg} shows that there are \NRA languages
that cannot be accepted by any \DRA\@.
Moreover, there is no computable bound for the number of registers $k$
which suffice to recognise a \kNRA 1 language by a \kDRA k (when such a number exists),
which follows from the following three observations:
\begin{enumerate}
  \item the \DRA membership problem is undecidable for \kNRA 1
  (\cref{thm:undecidabilityreg}),
  \item the problem of deciding equivalence of a given \kNRA 1 to a given \DRA is decidable by \cref{lem:OW04reg}
  and
  \item if an \kNRA 1 is equivalent to some \kDRA k then it is in fact equivalent to some \kDRA k
  with computably many control locations (by Lemma~\ref{thm:k-DRA-char} from the next section).
\end{enumerate}


\section{Decidability of \texorpdfstring{\kDRA{k}}{DRA-k} membership for \texorpdfstring{\kNRA{1}}{NRA-1}}%
\label{sec:upperboundreg}

%
In this section we prove our main decidability result for register automata,
which we now recall.
\thmDRAmemb*
The technical development of this section will also serve as a preparation for the more involved case of timed automata from \cref{sec:timedautomata,sec:inv}.
The key ingredient used in the proof of \cref{thm:kDRA:memb} is the following characterisation
of those \kNRA{1} languages which are also \kDRA{k} languages.
In particular, this characterisation provides a bound on the number of control locations of a \kDRA{k} equivalent to a given \kNRA{1} (if any exists).
%
%
\begin{lem}\label{thm:k-DRA-char}
	Let $A$ be a \kNRA{1} with $n$ control locations,
	and let $k\in\N$.
	The following conditions are equivalent:
	\begin{enumerate}
		\item%
		$\lang A = \lang B$ for some \kDRA{k} $B$.\label{pr1}
		\item For every data word $w$, there is $S \subseteq \A$ of size at most $k$ \st the left quotient
		$w^{-1} \lang A = \setof{v}{w \cdot v \in \lang A}$ is $S$-invariant.%
		\label{pr2}
		\item $\lang A = \lang B$ for some \kDRA{k} $B$ with at most
		$f(k,n)=(k+1)! \cdot 2^{n\cdot (k+1)}$ control locations.%
		\label{pr3}
	\end{enumerate}
\end{lem}



\noindent
Using the above lemma we derive a proof of \Cref{thm:kDRA:memb}:
\begin{proof}[Proof of \Cref{thm:kDRA:memb}]
	Given a \kNRA{1} $A$,
	the decision procedure enumerates all \kDRA{k} $B$ with at most $f(k, n)$ locations and checks whether $\lang A = \lang B$,
	using \Cref{lem:OW04reg}.
	If no such \kDRA{k} $B$ is found, the procedure answers negatively.
	Note that the fact that $A$ is \kNRA{1} is crucial here,
	since otherwise the equivalence check above would not be decidable.
	In fact, this is the only place where the one register restriction really plays a r\^ole:
	A generalisation of Lemma~\ref{thm:k-DRA-char} for $A$ being a $\kNRA l$ for $l \geq 1$
	could be stated and proved; however we prefer to avoid the additional notational complications of dealing with the general case
	since we need Lemma~\ref{thm:k-DRA-char} only in the case $l = 1$.
\end{proof}

\begin{rem}[Complexity]
	The decision procedure for \kNRA{1} invokes the \Ackermann~subroutine to check equivalence between a \kNRA 1 and a candidate \DRA\@.
	This is in a sense unavoidable,
	since we show in \Cref{thm:reg-lower-bound} that the \kDRA k membership problem is
	\Ackermann-hard for \kNRA{1}.
\end{rem}

The proof of Lemma~\ref{thm:k-DRA-char} is presented below in Sec.~\ref{sec:central-lemma}.
We remark that the lemma partially follows from the literature of automata theory in sets with atoms.
For instance, the most interesting implication~\eqref{pr2}$\limplies$\eqref{pr3},
even for $A$ an arbitrary $\kNRA l$ with $l \geq 1$,
minus the concrete bound $f(k, n)$ on control locations of $B$,
follows from the following known facts:
(a) $L(A)$ is recognised by a nondeterministic equivariant orbit-finite automaton~\cite[Theorem 5.11]{atombook},
(b) the assumption~\eqref{pr2} implies that the Myhill-Nerode equivalence of $L$ has orbit-finite index,
(c) thus by~\cite[Theorem 5.14]{atombook} $L(A)$ is recognised by a deterministic orbit-finite automaton $B'$, and
(d) we can construct from $B'$ some language-equivalent $\kDRA {k'}$ $B$.
However, this would not yield (1) the fact that we can even take $k' = k$
(i.e., the number of registers of $B$ can be taken to be the size $k$ of the supports $S$ in the assumption),
and (2) the concrete bound $f(k, n)$ on the number of control locations of $B$
(a computable such bound is necessary in the automata enumeration procedure in the proof of \Cref{thm:kDRA:memb}).
For these reasons, we provide a full proof of the lemma.

\subsection{Proof of Lemma~\ref{thm:k-DRA-char}}%
\label{sec:central-lemma}
%
Let us fix a \kNRA{1} $A = \tuple{\X, \Sigma, \L, \L_I, \L_F, \Delta}$ and $k\in\N$.
Let $n=\card \L$ be the number of control locations of $A$.
The implication~\eqref{pr3}$\limplies$\eqref{pr1} holds trivially.
The implication~\eqref{pr1}$\limplies$\eqref{pr2} holds just because every left quotient
$w^{-1} \lang A$ is the same as $w^{-1} \lang B$ by the assumption $\lang A = \lang {B}$ for a \kDRA{k} $B$,
and, since $B$ is deterministic,
the latter quotient $w^{-1} \lang B$ equals $\langa{B}{c}$ for some configuration $c = \tuple{p, \mu}$.
The latter is $\mu(\X)$-invariant by Lemma~\ref{lemma:NRA:invariance}(3),
and clearly $\card{\mu(\X)} \leq k$.
(Notice that $A$ did not play a r\^ole here.)

It thus remains to prove the implication~\eqref{pr2}$\limplies$\eqref{pr3},
which is the content of the rest of the section.
Assuming~\eqref{pr2}, we
are going to define a \kDRA{k} $B'$ with registers $\X = \set{\x_1, \dots, \x_{k}}$ and
with at most $f(k, n)$ locations such that $\lang {B'} = \lang A$.
We start from the transition system $\XX$ obtained by the finite powerset construction underlying
the determinisation of $A$.
Next, after a series of language-preserving transformations, we will obtain a transition system isomorphic
to the reachable part of $\semd{B'}$ for some \kDRA{k} $B'$.
As the last step, we extract from this deterministic transition system a syntactic definition of $B'$.
This is achievable due to the invariance properties witnessed by the transition systems in the course of the transformation.

\para{Macro-configurations}
For simplicity, we will abuse the notation and write $c = (p, a)$ for a configuration $c = (p, \set{\x_1 \mapsto a})$
of $A$, where $p \in \L$ and $a \in \A \cup \set{\bot}$.
A \emph{macro-configuration} is a (not necessarily finite)
set $X$ of configurations $(p, a)$ of $A$.
We use the notation
$\langa{A}{X} := \bigcup_{c \in X} \langa{A}{c}$.

Let $\succe{\sigma,a} X := \setof{c'\in\Conf{A}}{c \goesto{\sigma,a} c' \text{ for some } c \in X}$
be the set of successors of configurations in $X$
which can be reached by reading $\tuple{\sigma, a} \in \Sigma \times \A$.
We define a deterministic transition system $\XX$ consisting of
the macro-configurations reachable in the course of determinisation of $A$.
Let $\XX$ be the smallest set of macro-configurations and transitions such that

\begin{itemize}
\item $\XX$ contains the initial macro-configuration: $X_0 = \setof{(p, \bot)}{p\in \L_I} \in \XX$;
\item $\XX$ is closed under successor: for every $X\in\XX$ and $(\sigma,a)\in\Sigma\times\A$, there is a transition
$X \goesto {\sigma,a} \succe {\sigma,a} X$ in $\XX$.
\end{itemize}

\noindent
Due to the fact that $\semd A$ is finitely branching, i.e.~$\succe {\sigma,a} {\set{c}}$ is finite for every fixed $(\sigma,a)$,
all macro-configurations $X\in\XX$ are finite.
Let the final configurations of $\XX$ be $F_\XX = \setof{X\in\XX}{X\cap F \neq \emptyset}$ where $F \subseteq \semd A$ is the set of final configurations of $A$.
\begin{clm}\label{claim:eqlangAXreg}
$\langa A X = \langa \XX {X}$ for every $X\in\XX$. In particular $\lang A = \langa \XX {X_0}$.
\end{clm}
For a macro-configuration $X$ we write
$\macval{X} := \setof{a\in\A}{\exists p \st (p, a) \in X}$ to denote the set of atoms appearing in $X$.

\para{Pre-states}
By assumption~\eqref{pr2}, for every macro-configuration $X \in \XX$,
$\langa{A}{X}$ is $S$-invariant for some $S$ of size at most $k$,
but the macro-configuration $X$ itself needs not be $S$-invariant in general.
Indeed, a finite macro-configuration $X\in\XX$ is $S$-invariant if, and only if,
$\macval X \subseteq S$,
which is impossible in general when $X$ is arbitrarily large while
the size of $S$ is bounded (by $k$).
Intuitively, in order to assure $S$-invariance we will replace $X$ by its $S$-closure $\closure{S}{X}$
(recall Fact~\ref{fact:inv:clos}).

The \emph{least support} of a set with atoms $X$
is the least $S\subseteq\A$ \wrt set inclusion \st $X$ is $S$-invariant.
In the case of equality atoms every set with atoms has the least support,
which is moreover finite (see~\cite[Cor.~9.4]{LMCS14} or~\cite[Thm.~6.1]{atombook}).
By assumption, the least finite support of every macro-configuration $X$ in $\XX$ has size at most $k$.

A \emph{pre-state} is a pair $Y = (X, S)$,
where $X$ is a macro-state whose least finite support is $S$.
Thus $X$ is $S$-invariant which, together with the fact that $S$ has size at most $k$,
implies that there are only finitely many pre-states up to automorphism.
We define the deterministic transition system $\YY$ as the smallest set of pre-states
and transitions between them such that:

\begin{itemize}
	\item $\YY$ contains the initial pre-state: $Y_0 = (X_0, \emptyset) \in \YY$;
	\item $\YY$ is closed under the closure of successor:
	For every $(X, S) \in \YY$ and $(\sigma, a) \in \Sigma\times\A$,
	the pre-state $(X', S')$ is in $\YY$ together with transition $(X, S) \goesto {\sigma,a} (X', S')$,
	where $S'$ is the least finite support of the language $L' = (\sigma,a)^{-1} \langa{A}{X} = \langa{A}{\succe{\sigma,a}{X}}$,
	and $X' = \closure{S'}{\succe{\sigma,a}{X}}$.
\end{itemize}

\begin{exa}
\label{ex:ra:B}
   Suppose that $k=3$, a successor of some macro-configuration $X$ has the shape $\succe {\sigma, a_1} X = \set{(p,a_1),(q,a_1),(r,a_2),(s,a_3)}$ and
   the least finite support $S'$ of $L'$ is $\set{a_1, a_3}$, where $a_1, a_2, a_3 \in \A$ are pairwise-different.
   Then $X' = \set{(p,a_1), (q,a_1)}\cup\set{(r,a)\mid a \in \A \setminus \set{a_1, a_3}} \cup \set{(s,a_3)}$.
\end{exa}

\noindent
By assumption, $L'$ is $T$-invariant for some $T \subseteq \A$ with $\card{T} \leq k$.
Since $X$ is $S$-invariant, $L'$ is also $(S \cup \set{a})$-invariant.
By the least finite support property of equality atoms,
finite supports are closed under intersection, and hence $S' \subseteq (S \cup \set{a}) \cap T$, which implies
$\card {S'}\leq k$.

By Lemma~\ref{lemma:NRA:invariance} we deduce:
\begin{restatable}[Invariance of $\YY$]{clm}{claimEquivY}%
\label{claim:equivYreg}
For every two transitions
\[
    (X_1, S_1) \goesto {\sigma,a_1} (X'_1, S'_1)
    \qquad\text{and}\qquad
    (X_2, S_2) \goesto {\sigma, a_2} (X'_2, S'_2)
\]
in $\YY$
and an automorphism $\pi$, if $\pi(X_1) = X_2$ and $\pi(S_1) = S_2$ and $\pi(a_1) = a_2$, then we have
$\pi(X'_1) = X'_2$ and $\pi(S'_1) = S'_2$.
\end{restatable}
Let the final configurations of $\YY$ be $F_\YY = \setof{(X,S)\in\YY}{X\cap \F \neq \emptyset}$.
By induction on the length of data words it is easy to show:
\begin{clm}\label{claim:eqlangXYreg}
	$\langa \XX {X_0} = \langa \YY {Y_0}$.
\end{clm}
%

\para{States}
We now introduce \emph{states}, which are designed to be in one-to-one correspondence
with configurations of the forthcoming \kDRA{k} $B'$.
Intuitively, a state differs from a pre-state $(X, S)$ only by allocating the values from
(some superset of) $S$ into $k$ registers.
Thus, while a pre-state contains a set $S$, the corresponding state contains a register assignment
$\mu : \A_\bot^\X$ with image $\mu(\X) \supseteq S$.

Let $\X = \set{\x_1, \dots, \x_k}$ be a set of $k$ registers.
A \emph{state} is a pair $Z = (X, \mu)$,
where $X$
is a macro-configuration,
$\mu : \A_\bot^\X$ is a register assignment, and $X$ is $\mu(\X)$-invariant.
Thus every state $(X, \mu)$ determines uniquely a corresponding pre-state
$\tau(X, \mu) = (X, S)$ where $S\subseteq\mu(\X)$ is the least finite support of $X$.

\begin{exa}%
\label{ex:ra:C}
We continue Example~\ref{ex:ra:B}. States corresponding to the pre-state $(X', S')$ feature the macro-configuration $X'$, but can have different register valuations. One of them is $(X', \mu')$, where $\mu' = \set{\x_1 \mapsto a_1, \x_2 \mapsto a_3, \x_3 \mapsto a_1}$.
\end{exa}

We now define a deterministic transition system $\ZZ$.
Its states are all those $(X,\mu)$ satisfying $\tau(X,\mu) \in \YY$,
and transitions are determined as follows:
$\ZZ$ contains a transition $(X, \mu) \goesto {\sigma, a} (X', \mu')$ if $\YY$ contains the corresponding
transition
$\tau(X, \mu) \goesto {\sigma, a} \tau(X', \mu') = (X', S')$,
and $\mu' = \extend \Y a \mu$, where 
\begin{align}\label{eq:mureg}
\Y \ = \ \setof{\x_i\in\X}{\mu(\x_i) \notin S' \text{ or } \mu(\x_i) = \mu(\x_j) \text{ for some } j>i}.
\end{align}
%
%
The equation~\eqref{eq:mureg} defines a deterministic update policy\footnote{There are in general many correct deterministic update policies, but for our purposes it suffices to define one such deterministic update policy.} of the register assignment $\mu$
that amounts to updating with the current input atom $a$ all registers
$\x_i$ whose value is either no longer needed (because $\mu(\x_i) \notin S'$), or
is shared with some other register $x_j$, for $j>i$ and is thus redundant.
It is easy to see that the above register update policy guarantees that
$S' \subseteq \mu'(\X)\subseteq S'\cup\set{a}$.
%
Using Claim~\ref{claim:equivYreg} 
we derive:
\begin{restatable}[Invariance of $\ZZ$]{clm}{claimEquivZ}%
\label{claim:equivZreg}
For every two transitions
\[
    (X_1, \mu_1) \goesto {\sigma,a_1} (X'_1, \mu'_1)
    \qquad\text{and}\qquad
    (X_2, \mu_2) \goesto {\sigma,a_2} (X'_2, \mu'_2)
\]
in $\ZZ$
and an automorphism $\pi$, if $\pi(X_1) = X_2$ and $\pi{\circ}{\mu_1} = \mu_2$ and $\pi(a_1) = a_2$, then we have
$\pi(X'_1) = X'_2$ and $\pi{\circ}{\mu'_1} = \mu'_2$.
\end{restatable}
Let the initial state be $Z_0 = (X_0, \lambda \x . \bot)$,  and let final states be $F_\ZZ = \setof{(X, \mu)\in \ZZ}{X \cap F \neq \emptyset}$.
By induction on the length of data words one proves:
\begin{clm}\label{claim:eqlangYZreg}
$\langa \YY {Y_0} = \langa \ZZ {Z_0}$.
\end{clm}
In the sequel we restrict $\ZZ$ to states reachable from $Z_0$.

\para{Orbits of states}
Recall that the action of automorphisms on macro-configurations and reset-point assignments is extended
to states as $\pi(X, \mu) = (\pi(X), \pi{\circ}\mu)$, and
that the orbit of a state $Z$ is defined as $\orbit {}  {Z} = \setof{\pi(Z)}{\pi \in \Aut {} \A}$.

While a state is designed to correspond to a configuration of the forthcoming \kDRA{k} $B'$,
its orbit is designed to play the r\^ole of control location of $B'$.
We therefore need to prove that
the set of orbits $\setof{\orbit {} Z}{Z\in\ZZ}$
is finite and its size is bounded by $f(k, n)$.

Let $M_k$ denote the number of orbits of register valuations $\A_\bot^\X$, which is the same as the number of
orbits of $k$-tuples $\A_\bot^k$.
In case of equality atoms we have $M_k \leq (k+1)!$ Indeed, 
at the first position there are two possibilities, $\bot$ or an atom;
at the second position there are at most three possibilities: $\bot$, the same atom at the one at the first position, or a fresh atom;
and so on, until the last position where there are at most $k+1$ possibilities.

Every $l$-element subset $S = \set{a_1, a_2, \ldots, a_l} \subseteq \A$ of atoms induces, in the case of equality atoms,
exactly $l + 1$ different $S$-orbits of atoms:
\[
\orbity S \ = \ \setof{\orbit S a}{a\in\A} \ = \ \set{\set{a_1}, \set{a_2}, \ldots, \set{a_l}, \A - S}.
\]
Therefore, each $S\subseteq \A$ of size at most $k$ induces at most $N_k := k+1$ different $S$-orbits of atoms.

Consider a state $Z = (X, \mu)$ and let $S=\mu(\X)$.
We define the characteristic function
$
\charact Z : \orbity S \to \powerset \L
$
as follows:
\[
\charact Z(o) = \setof{l\in \L}{(l, a)\in X \text{ for some }a\in o}.
\]
Since $X$ is $S$-invariant, the choice of the atom $a\in o$ is irrelevant, and we conclude:
\begin{clm}\label{claim:det}
	Every state $Z = (X,\mu)$  is uniquely determined by its register valuation $\mu$ and by the characteristic function $\charact Z$.
\end{clm}

\begin{clm}\label{claim:nrorb}
The size of $\setof{\orbit {} Z}{Z\in\ZZ}$
is at most $M_k \cdot 2^{n\cdot N_k}$.
\end{clm}
\begin{proof}
We show that there are at most $M_k \cdot (2^n)^{N_k}$ different orbits of states.
Consider two states $Z=(X, \mu)$ and $Z'=(X', \mu)$
and let $S = \mu(\X)$ and $S' = \mu'(\X)$.
Suppose that the register valuations $\mu$ and $\mu'$ are in the same orbit:
$\pi\circ \mu = \mu'$ for some automorphism $\pi$.
Thus $\pi(S) = S'$, and moreover $\pi$ induces a bijection $\widetilde \pi$ between $\orbity S$ and $\orbity {S'}$.
Note that, once $S$ is fixed, there are at most $(2^n)^{N_k}$ possible characteristic functions of $Z$,
and likewise for $S'$ and $Z'$.
Supposing further that the characteristic functions agree, i.e., satisfy $\charact Z = \charact {Z'}\circ \widetilde \pi$,
using Claim~\ref{claim:det} we derive $\pi(Z) = Z'$, i.e., $Z$ and $Z'$ are in the same orbit.
Therefore, since the number of orbits of register valuations $\mu, \mu'$ is at most $M_k$, and for each such orbit the
number of different characteristic functions is at most $(2^n)^{N_k}$,
the number of different orbits of states is bounded as required.
\end{proof}

For future use we observe that
every state is uniquely determined by its register valuation and its orbit:
	\begin{clm}%
	\label{claim:lastreg}
	Let $Z = (X, \mu)$ and $Z' = (X', \mu)$ be two states in $\ZZ$ with the same register valuation.
	If $\pi(X) = X'$ and $\pi{\circ}\mu = \mu$ for some automorphism $\pi$ then $X = X'$.
\end{clm}
\begin{proof}
Indeed, $X$ is $\mu(X)$-invariant and hence $\pi(X) = X$, which implies $X = X'$.
\end{proof}

In the terminology of automata in sets with atoms,
we have proved that $\ZZ$ is a \emph{deterministic orbit-finite automaton} (c.f.~\cite[Sec.~5.2]{atombook} for a definition),
with the concrete bound on the number of orbits given by Lemma~\ref{claim:nrorb}.

\para{Construction of the \DRA}
As the last step we define a \kDRA {k} $B' = \tuple{\X, \Sigma, \L', \set{o_0}, \L'_F, \Delta'}$
such that the reachable part of $\semd {B'}$ is isomorphic to $\ZZ$.
Let locations $\L' = \setof{\orbit {} Z}{Z\in\ZZ}$ be the orbits of states from $\ZZ$,
the initial location be the orbit $o_0$ of $Z_0$, and final locations
$\L'_F = \setof{\orbit{}{Z}}{Z\in F_\ZZ}$ be orbits of final states.
Let each transition $Z = (X, \mu) \goesto{\sigma, a} (X', \mu') = Z'$ in $\ZZ$
induce a transition rule in $B'$
\begin{align}\label{eq:trofBreg}
	\transition o \sigma \psi \Y {o'} \in \Delta'
\end{align}
where $o = \orbit{}{Z}$, $o' = \orbit{}{Z'}$, $\Y = \setof{\x\in\X}{\mu'(\x) = a}$,
and the constraint
$\psi(\x_1, \ldots, \x_k, \y)$ defines the orbit of $\tuple{\mu(\x_1), \dots, \mu(\x_k), a}$
(here we rely on Claim~\ref{claim:def-inv}).
We argue that the automaton $B'$ is deterministic:
\begin{clm}
	Suppose that $B'$ has two transition rules $o \goesto{\sigma, \psi_1, \Y_1} o'_1$ and
	$o \goesto{\sigma, \psi_2, \Y_2} o'_2$
	with the same source location $o$ and jointly satisfiable constraints ($\semlog{\psi_1 \land \psi_2} \neq \emptyset$).
	Then the target locations are equal ($o'_1 = o'_2$), and the same registers are updated ($\Y_1 = \Y_2$).
\end{clm}
\begin{proof}
	Since the constraints are jointly satisfiable, both transition rules are enabled in some configuration $c = (o, \mu)$
	and for some input atom $a \in \A$.
	%
		By Claim~\ref{claim:lastreg},
	$c$ determines a corresponding state $Z = (X, \mu)$ with $o = \orbit {} Z$ and, since the system $\ZZ$ is deterministic,  both transition
	rules are induced by a common transition $(X, \mu) \goesto{\sigma,a} (X', \mu')$ in $\ZZ$. This in turn implies
	$o_1 = o_2$ and $\Y_1 = \Y_2$, as required.
\end{proof}

\begin{clm}\label{claim:eqlangCBreg}
	$\ZZ$ is isomorphic to the reachable part of $\semd {B'}$.
\end{clm}
\begin{proof}
	For a state $Z = (X, \mu)$,
	let $\confof Z = (\orbit{}{Z}, \mu)$.
	Let $\ZZ'$ denote the reachable part of $\semd {B'}$.
	By Claim~\ref{claim:lastreg},
	the mapping $\confof \_$ is a bijection between $\ZZ$ and its image $\confof{\ZZ} \subseteq \semd {B'}$.
	We aim at proving $\confof \ZZ = \ZZ'$.

	By the very definition~\eqref{eq:trofBreg}, the image
	$\confof \ZZ$ is a subsystem of $\ZZ'$: $\confof \ZZ \subseteq \ZZ'$.
	For the converse inclusion,
	recall that $\ZZ$ is total: for every
	$(\sigma_1, a_1) \dots (\sigma_n, a_n) \in (\Sigma \times \A)^*$,
	there is a sequence of transitions $(X_0, \mu_0) \goesto {\sigma_1, a_1} \cdots \goesto{\sigma_n, a_n} (X_n, \mu_n)$ in $\ZZ$.
	Therefore $\confof \ZZ$ is total too and, since $\ZZ'$ is deterministic and reachable, the subsystem $\confof \ZZ\subseteq \ZZ'$
	necessarily equals $\ZZ'$.
%
%
\end{proof}

Claims~\ref{claim:eqlangAXreg},~\ref{claim:eqlangXYreg},~\ref{claim:eqlangYZreg},
and~\ref{claim:eqlangCBreg} jointly imply $\lang A = \lang {B'}$, which
completes the proof of Lemma~\ref{thm:k-DRA-char}.

\subsection{Other atoms}%
\label{sec:OtherAtoms}

The proof of  \Cref{thm:kDRA:memb} straightforwardly generalises to any relational structure of atoms $\A$ satisfying the following conditions:
\begin{itemize}
\item $\A$ is homogeneous~\cite{Fraissebook};
\item $\A$ preserves well-quasi orders (\wqo): finite induced substructures of $\A$ labelled by elements of an arbitrary \wqo,
ordered by label-respecting embedding, are again a \wqo (for details we refer the reader to~\cite[item (A3), Sect.5]{Lasota16});
\item $\A$ is effective: it is decidable, if a given finite structure over the vocabulary of $\A$ is an induced substructure thereof;
\item $\A$ has the least finite support property.
\end{itemize}
As an example, the structure of densely ordered atoms $\A = (\R, \leq)$ satisfies all the conditions and hence
\Cref{thm:kDRA:memb} holds for register automata over this structure of atoms.

We briefly discuss the adjustments needed.
The syntax of constraints~\eqref{eq:constr} is extended by adding atomic constraints for all relations in $\A$, and
Claim~\ref{claim:def-inv} holds by homogeneity of $\A$.
The decision procedure checking equivalence of a \kNRA 1 and a \DRA from \cref{lem:OW04reg},
invoked in the proof of \Cref{thm:kDRA:memb},
works assuming that $\A$ preserves \wqo and is effective.
The least finite support assumption is required in the definition of pre-states.
Finally, again due to homogeneity of $\A$ the bounds $M_k$ and $N_k$ used in Claim~\ref{claim:nrorb} are finite
(but dependent on $\A$).
%


\section{Timed automata}\label{sec:timedautomata}

The register-based model which is the closest to timed automata
is based on timed atoms $\tuple{\R, \leq, {+1}}$~\cite{BL12}.
However, due to the syntactic and semantic restrictions which are traditionally imposed on timed automata
(such as monotonicity of time, nonnegative timestamps, the special status of the initial timestamp 0, and the concrete syntax of transition constraints),
the latter can be seen only as as a strict subclass of the corresponding register model.
For this reason, we present our results on the deterministic membership probem in the syntax and semantics of timed automata.



\para{Timed words and languages}
Fix a finite alphabet $\Sigma$.
Let $\R$ and $\Rnonnegpos$ denote the reals, resp., the nonnegative reals%
\footnote{Equivalently, the rationals $\mathbb Q$ could be considered in place of reals.}.
%
Timed words are obtained by instantiating data words to timed atoms,
and imposing additional conditions: non-negativeness
and monotonicity:
A \emph{timed word} over $\Sigma$ is any sequence of the form
\begin{align}\label{eq:newtw}
	w \ = \ 
	(\sigma_1, t_1)\, \dots \, (\sigma_n, t_n) \ \in \ (\Sigma \times\Rnonneg)^*
\end{align}
%
which is \emph{monotonic},
in the sense that the timed atoms (called \emph{timestamps} henceforth)
$t_i$'s satisfy $0 \leq t_1 \leq t_2 \leq \dots \leq t_n$. 
%
For $w$ a timed word as in~\eqref{eq:newtw} and an increment $\delta \in \Rnonneg$,
let $w + \delta = (\sigma_1, t_1 + \delta)\, \dots \, (\sigma_n, t_n + \delta)$
be the timed word obtained from $w$ by increasing all timestamps by $\delta$.
Let $\timedwords{\Sigma}$ be the set of all timed words over $\Sigma$,
and let $\timedwordsafter{\Sigma}{t}$ be, for $t\in \Rnonneg$,
the set of timed words with $t_1 \geq t$. 
A \emph{timed language} is any subset of $\timedwords{\Sigma}$.


The concatenation $w \cdot v$ of two timed words $w$ and $v$
is defined only when the first timestamp of $v$ is greater or equal than the last timestamp of $w$.
%
Using this partial operation, we define,
for a timed word $w\in\timedwords{\Sigma}$ and a timed language $L\subseteq \timedwords\Sigma$, the left quotient
$w^{-1} L := \setof{v\in\timedwords\Sigma}{w \cdot v \in L}$.
Clearly $w^{-1} L \subseteq \timedwordsafter \Sigma {t_n}$. 

\para{Clock constraints and regions}
Let $\X = \set{\x_1, \dots, \x_k}$ be a finite set of clocks.
A \emph{clock valuation} is a function $\mu \in \Rnonnegpos^\X$
assigning a non-negative real number $\mu(\x)$ to every clock $\x \in \X$.
A \emph{clock constraint} is a quantifier-free formula of the form
\begin{align*}
    \varphi,\psi \ ::\equiv\ \true \sep \false \sep \x_i - \x_j \sim z \sep \x_i \sim z \sep \neg \varphi \sep \varphi \land \psi \sep \varphi \lor \psi,
\end{align*}
where ``$\sim$'' is a comparison operator in $\set{=, <, \leq, >, \geq}$
and $z \in \Z$.
A clock valuation $\mu$ satisfies a constraint $\varphi$, written $\mu \models \varphi$,
if interpreting each clock $\x_i$ by $\mu(\x_i)$ makes $\varphi$ a tautology.
Let $\sem \varphi$ be the set of clock valuations $\mu \in \Rnonnegpos^\X$ \st $\mu \models \varphi$.
If $\sem \varphi$ is non-empty and it does not strictly include a non-empty $\sem \psi \subsetneq \sem \varphi$ for some constraint $\psi$,
then we say that it is a \emph{region}.
For example, for $\X = \set{\x, \y}$
we have that $r_0 = \sem {1 < \x < 2 \land \y = 3}$ is a region,
while $r_1 = \sem {1 < \x \leq 2 \land \y = 3}$ is not.
Nonetheless, the latter partitions into two regions $r_1 = r_0 \cup \sem{\x = 2 \land \y = 3}$,
and we will later see that this is a general phenomenon.
%
%
A \emph{$k, m$-region} is a region $\semlog{\varphi}$
where $\varphi$ has $k$ clocks and absolute value of maximal constant bounded by $m$.
%
For instance, the clock constraint
$1 < \x_1 < 2 \;\wedge\; 4 < \x_2 < 5 \;\wedge\; \x_2 - \x_1 < 3$
defines a $2,5$-region consisting of an open triangle with nodes $(1, 4)$, $(2, 4)$ and $(2, 5)$.
A region $\sem \varphi$ is \emph{bounded} if it is bounded as a subset of $\Rnonnegpos^\X$ in the classical sense,
i.e., there exists $M \in \Rnonnegpos$ \st $\sem \varphi \subseteq [0, M]^\X$.

\para{Timed automata}

A (nondeterministic) \emph{timed automaton} is a tuple $A = \automaton$,
where $\X$ is a finite set of clocks,
$\Sigma$ is a finite input alphabet,
$\L$ is a finite set of control locations,
$\I, \F \subseteq \L$ are the subsets of initial, resp., final, control locations,
and $\Delta$ is a finite set of transition rules of the form
\begin{align}\label{eq:trans-rule}
	\transition{p}{\sigma}{\varphi}{\Y}{q}
\end{align}
with $p, q \in \L$ control locations, $\sigma \in \Sigma$,  
$\varphi$ a clock constraint to be tested,
and $\Y \subseteq \X$ the set of clocks to be reset.
We write \NTA for the class of all nondeterministic timed automata, \kNTA{k} when the number $k$ of clocks is fixed,
\mNTA{m} when the bound $m$ on constants is fixed,
and \kmNTA{k}{m} when both $k$ and $m$ are fixed.

An \mNTA m $A$ is \emph{always resetting} if every transition rule resets some clock
($\Y \neq \emptyset$ in~\eqref{eq:trans-rule}),
and \emph{greedily resetting}
if, for every clock $\x$,
whenever $\varphi$ implies that the value of $\x$ belongs to $\set{0, \ldots, m} \cup (m, \infty)$, then $\x \in \Y$.
Intuitively, a greedily reseting automaton resets every clock whose value is either an integer,
or exceeds the maximal constant $m$.

\para{Reset-point semantics}

We introduce a semantics based on reset points
instead of clock valuations.
A \emph{reset-point assignment} is a function $\mu \in \Rnonneg^\X$
storing, for each clock $\x \in \X$, the timestamp $\mu(\x)$ when $\x$ was last reset.
Reset-point assignments and clock valuations have the same type $\Rnonneg^\X$,
however we find it technically more convenient to work with reset points than with clock valuations.
The reset-point semantics has already appeared in the literature on timed automata~\cite{Fribourg:1998}
and it is the foundation of the related model of timed-register automata~\cite{BL12}.

A \emph{configuration} of an \NTA $A$ is a tuple $\tuple{p, \mu, \now}$
consisting of a control location $p \in \L$,
a reset-point assignment $\mu \in \Rnonneg^\X$, and a ``now'' timestamp $\now \in \Rnonneg$
satisfying $\mu(\x) \leq \now$ for all clocks $\x\in \X$.
Intuitively, $\now$ is the last timestamp seen in the input
and, for every clock $\x$,
$\mu(\x)$ stores the timestamp of the last reset of $\x$.
A configuration is \emph{initial} if $p$ is so, $\now = 0$, and $\mu(\x) = 0$ for all clocks $\x$,
and it is \emph{final} if $p$ is so
(without any further restriction on $\mu$ or $\now$).
%
For a set of clocks $\Y \subseteq \X$ and a timestamp $u\in \Rnonneg$, let $\extend \Y u \mu$
be the assignment which is $u$ on $\Y$ and agrees with $\mu$ on $\X \setminus \Y$.
A reset-point assignment $\mu$ together with $\now$ induces the clock valuation $\assigntoval \mu \now$ defined as
$(\assigntoval \mu \now)(\x) = \now - \mu(\x)$ for all clocks $\x\in\X$.

Every transition rule~\eqref{eq:trans-rule}
induces a \emph{transition} between configurations
\[\tuple {p, \mu, \now} \goesto {\sigma,t} \tuple {q, \nu, t}\]
labelled by $(\sigma,t)\in\Sigma\times\Rnonneg$ 
whenever 
%
    $t\geq \now$, 
    $\assigntoval \mu t \models \varphi$, and 
    $\nu = \extend \Y t \mu$.
%
\noindent
%
The \emph{timed transition system} induced by $A$
is $\tuple{\Conf A, \goesto {}, F}$,
where $\Conf A$ is the set of configurations,
${\goesto{}} \subseteq \Conf A \times \Sigma \times \Rnonnegpos \times \Conf A$
is as defined above, and
$F \subseteq \Conf A$ is the set of final configurations.
Since there is no danger of confusion,
we use $\semd A$ to denote either the timed transition system above, or its domain.
A \emph{run} of $A$ \emph{over} a timed word $w$ as in~\eqref{eq:newtw}
\emph{starting} in configuration $\tuple {p, \mu, t_0}$
and \emph{ending} in configuration $\tuple {q, \nu, t_n}$
is a path $\rho$ in $\semd A$ 
of the form
%
$	\rho  =  \tuple {p, \mu, t_0}
		\goesto {\sigma_1,t_1} \ 
		\dots
		\goesto {\sigma_n, t_n} \ 
		\tuple{q, \nu, t_{n}}$. 
%
%
The run $\rho$ is accepting if its last configuration satisfies $\tuple {q, \nu, t_{n}}\in F$.
The language \emph{recognised} by configuration $(p, \mu, \now)$ is defined as:
\begin{align*}
	\langtsa {\semd A} {p, \mu, \now} = \setof{w\in\timedwords{\Sigma}}{\semd A \text{ has an accepting run over } w \text{ starting in } \tuple{p, \mu, \now}}. 
\end{align*}
%
Clearly $\langtsa {\semd A} {p, \mu, \now} \subseteq \timedwordsafter \Sigma \now$.
We write $\langtsa {A} c$ instead of $\langtsa {\semd A} c$.
The language recognised by the automaton $A$ is
$\langts A = \bigcup_{c \text{ initial}} \langtsa A {c}$.

A configuration is \emph{reachable} if it is the ending configuration in a run starting in an initial configuration.
In an always resetting \mNTA m,
every reachable configuration $(p, \mu, \now)$ satisfies $\now \in \mu(\X)$,
where $\mu(\X) = \setof{\mu(\x)}{\x\in\X}$;
in a greedily resetting one,
(1) $(p, \mu, \now)$ has \emph{$m$-bounded span},
in the sense that $\mu(\X) \subseteq \lopen{\now-m, \now}$,
and moreover (2) any two clocks $\x, \y$ with integer difference
$\mu(\x) - \mu(\y) \in \Z$ are actually equal $\mu(\x) = \mu(\y)$.
Condition (2) follows from the fact that if $\x, \y$ have integer difference and $\y$ was reset last,
then $\x$ was itself an integer when this happened,
and in fact they were both reset together in a greedily resetting automaton.

\para{Deterministic timed automata}
A timed automaton $A$ 
is \emph{deterministic} if it has exactly one initial location
and, for every two rules
$\transition{p}{\sigma}{\varphi}{\Y}{q}, \transition{p}{\sigma'}{\varphi'}{\Y'}{q'} \in \Delta$,
if $\sigma = \sigma'$
and the two transition constraints are jointly satisfiable $\semlog{\varphi \land \varphi'} \neq \emptyset$,
then $\Y = \Y'$ and $q = q'$.
Hence $A$ has at most one run over every timed word $w$.
A \DTA can be easily transformed to a \emph{total} one, where for every location $p\in\L$ and $\sigma\in\Sigma$,
the sets defined by clock constraints
$\setof{\semlog \varphi}{\prettyexists{\Y, q}{\transition{p}{\sigma}{\varphi}{\Y}{q} \in \Delta}}$
are a partition of $\Rnonnegpos^\X$.
Thus, a total \DTA has exactly one run over every timed word $w$.
%
We write \DTA for the class of deterministic timed automata,
and \kDTA{k}, \mDTA{m}, and \kmDTA{k}{m} for the respective subclasses thereof.
A timed language is called \NTA language, \DTA language, etc.,
if it is recognised by a timed automaton of the respective type.
%

\begin{exa}\label{example:L1}
    This is a timed analog of Example~\ref{example:L1reg}.
    Let $\Sigma = \set \sigma$ be a unary alphabet.
	As an example of a timed language $L$ recognised by a \kNTA{1},
	but not by any \DTA,
    consider the set of non-negative timed words of the form
		$(\sigma, t_1) \cdots (\sigma, t_n)$
	where $t_n - t_i = 1$ for some $1\leq i < n$.
	The language $L$ is recognised by the \kNTA{1} $A = \automaton$
	with a single clock $\X = \set \x$ and three locations $\L = \set{p, q, r}$,
	of which $\I = \set p$ is initial and $\F = \set r$ is final, and transition rules
	\begin{align*}
		& \transition{p}{\sigma}{\true}{\emptyset}{p} \qquad
		\transition{p}{\sigma}{\true}{\set{\x}}{q} \qquad
		\transition{q}{\sigma}{\x<1}{\emptyset}{q}  \qquad
		\transition{q}{\sigma}{\x=1}{\emptyset}{r}.
	\end{align*}
	Intuitively, in $p$ the automaton waits until it guesses that the next input will be $(\sigma, t_i)$,
	at which point it moves to $q$ by resetting the clock.
	%
	From $q$, the automaton can accept by going to $r$ only if exactly one time unit elapsed since $(\sigma, t_i)$ was read.
	The language $L$ is not recognised by any \DTA
	since, intuitively, any deterministic acceptor needs to store unboundedly many timestamps $t_i$'s.
\end{exa}

\para{Deterministic membership problems}

The decision problems for \NTA we are interested in are analogous to the ones for \NRA,
but additionally a bound on the maximal constant appearing in a \DTA may be specified as a parameter.
Let $\mathcal X$ be a subclass of \NTA\@.
We are interested in the following decision problem.

\decision{
$\mathcal X$ membership problem}
{A timed automaton $A\in$ \NTA\@.}
{Does there exist a 
$B\in \mathcal X$ \st $\lang A = \lang B$?}

In the rest of the paper, we study the decidability status of the $\mathcal X$ membership problem
where $\mathcal X$ ranges over \DTA, \kDTA k (for every fixed number of clocks $k$),
\mDTA m (for every maximal constant $m$),
and \kmDTA k m (when both clocks $k$ and maximal constant $m$ are fixed).
Example~\ref{example:L1} shows that there are \NTA languages
which cannot be accepted by any \DTA\@.
Moreover, similarly as in case of \NRA, there is no computable bound for the number of clocks $k$
which suffice to recognise a \kNTA 1 language by a \kDTA k (when such a number exists).


\section{Invariance of timed automata}%
\label{sec:inv}

A fundamental tool used below 
is invariance properties of timed languages recognised by \NTA with respect
to timed automorphisms.
In this section we establish these properties, as extension of analogous properties of register automata.
We also prove a timed analog of the least support property, and relate regions to orbits of configurations.

%

Recall that timed automorphisms are monotonic bijections $\R\to\R$ that preserve integer differences.
A timed automorphism $\pi$ acts on input letters in $\Sigma$ as the identity, $\pi(a) = a$,
and is extended point-wise to
timed words $\pi((\sigma_1, t_1) \dots (\sigma_n, t_n)) = (\sigma_0, \pi(t_1)) \dots (\sigma_n, \pi(t_n))$,
configurations $\pi(p, \mu, \now) = (p, \pi{\circ}\mu, \pi(\now))$,
transitions
$\pi(c \goesto {\sigma, t} c') = \pi(c) \goesto {\sigma, \pi(t)} \pi(c')$,
and sets $X$ thereof $\pi(X) = \setof{\pi(x)}{x\in X}$.
%
\begin{rem}
In considerations about timed automata we restrict to nonnegative reals, while
  a timed automorphism $\pi$ can in general take a nonnegative real $t\geq 0$ to a negative one.
  %
  In the sequel whenever we write $\pi(x)$, for $x$ being any object like a timestamp, a configuration, a timed word, etc.,
  we always implicitly assume that $\pi$ is well-defined on $x$, i.e., yields a timestamp, a configuration, a timed word, etc.
  In other words, for invariance properties we restrict to those timed automorphisms that preserve nonnegativeness
  of all the involved timestamps.
\end{rem}

Let $S \subseteq \Rnonneg$.
An \emph{$S$-timed automorphism} is a timed automorphism \st $\pi(t) = t$ for all $t\in S$.
Let $\Pi_S$ denote the set of all $S$-timed automorphisms, and let $\Pi = \Pi_\emptyset$.
A set $X$ is \emph{$S$-invariant}
if $\pi(X) = X$ for every $\pi \in \Pi_S$;
equivalently,
for every $\pi \in \Pi_S$,
$x\in X$ if, and only if $\pi(x)\in X$.
%
A set $X$ is \emph{invariant} if it is $S$-invariant with $S = \emptyset$.
The following three facts express some basic invariance properties.
%
\begin{restatable}{fact}{invarianceTrans}
  \label{fact:equivariant:trans}
  The timed transition system $\semd A$ is invariant.
\end{restatable}
\begin{proof}
  Suppose $c = (p, \mu, \now) \goesto {\sigma,t} (p', \mu', t) = c'$ due to some transition rule of $A$
  whose clock constraint $\varphi$ compares values of clocks
  $\x$, i.e., the differences $t - \mu(\x)$, to integers.
  Since a timed automorphism $\pi$ preserves integer distances, the same clock constraint is satisfied
  in $\pi(c) = (p, \pi{\circ}{\mu}, \pi(\now))$, and therefore the same transition rule is applicable yielding the transition
  $(p, \pi{\circ}{\mu}, \pi(\now)) \goesto {\sigma,\pi(t)} (p, \pi{\circ}{\mu'}, \pi(t)) = \pi(c')$.
\end{proof}

\noindent
By unrolling the definition of invariance in the previous fact,
we obtain that the set of configurations is invariant,
the set of transitions ${\goesto{}}$ is invariant,
and that the set of final configurations $F$ is invariant.
\begin{restatable}[Invariance of the language semantics]{fact}{factEquivLang}%
    \label{fact:equivariant:lang}
    %
    The function $c \mapsto \langtsa A c$ from $\semd A$ to languages is invariant,
    i.e., for all timed automorphisms $\pi$,
    $\langtsa A {\pi(c)} = \pi(\langtsa A c)$.
\end{restatable}
\begin{proof}
Consider a timed automorphism $\pi$ and an accepting run of $A$ over a timed word
$w = (\sigma_1, t_1)  \dots (\sigma_n, t_n) \in \timedwordsafter \Sigma \now$ starting in $c = (p, \mu, \now)$:
\begin{align*}
           (p, \mu, \now) 
     \goesto {\sigma_1, t_1} \,
     \cdots
                    \goesto{\sigma_n, t_n} \,
                        (q, \nu, t_n),
\end{align*}
After $\sigma_i$ is read, the value of each clock is either the difference $t_i - \mu(\x)$ for some $1\leq i \leq n$ and clock $\x\in\X$,
or the difference $t_i - t_j$ for some $1 \leq j \leq i$.
Likewise is the difference of values of any two clocks.
Thus clock constraints of transition rules used in the run compare these differences to integers.
As timed automorphism $\pi$ preserves integer differences, by executing the same sequence of transition rules we obtain
the run over $\pi(w)$ starting in $\pi(c) = (p, \pi{\circ}\mu, \pi(\now))$: 
\[
           (p, \pi{\circ}\mu, \pi(\now)) 
   \goesto {\sigma_1, \pi(t_1)} \,
   \cdots
                    \goesto{\sigma_n, \pi(t_n)} \,
                        (q, \pi{\circ}\nu, \pi(t_n)),
\]
also accepting as it ends in the same location $q$.
As $w\in\timedwords \Sigma$ can be chosen arbitrarily, we have thus proved one of inclusions, namely
\[
\pi(\langtsa A {p, \mu, \now}) \ \subseteq \ \langtsa A {p, \pi{\circ}\mu, \pi(\now)}.
\]
The other inclusion follows from the latter one applied to $\pi^{-1}$ and $\langtsa A {p, \pi{\circ}\mu, \pi(\now)}$:
\[
\pi^{-1}(\langtsa A {p, \pi{\circ}\mu, \pi(\now)}) \ \subseteq  \ \langtsa A {p, \pi^{-1}{\circ}\pi{\circ}\mu, \pi^{-1}(\pi(\now))}
\ = \ \langtsa A {p, \mu, \now}.
\]
The two implications prove the equality.
\end{proof}

%
\begin{restatable}[Invariance of the language of a configuration]{fact}{invarianceLang}%
  \label{fact:invariantbase}%
  \label{fact:invariantalways}
  The language $\langtsa A {p, \mu, \now}$ is
  $(\clockval c \cup\set{\now})$-invariant.  
  Moreover, if $A$ is always resetting,
  then $\langtsa A {p, \mu, \now}$ is
  $\clockval c$-invariant.
\end{restatable}
\begin{proof}
  This is a direct consequence of the invariance of semantics.
  Indeed, for every $(\clockval c\cup\set{\now})$-timed automorphism $\pi$ the configurations $c = (p, \mu, \now)$ and
  $\pi(c) = (p, \pi{\circ}\mu, \pi(\now))$ are equal,
  hence their languages $\langtsa A {c}$ and
  $\langtsa A {\pi(c)}$, the latter equal to $\pi(\langtsa A c)$ by Fact~\ref{fact:equivariant:lang}, are equal too.
  Thus, $L = \pi(L)$.
  Finally, if $A$ is always resetting, then $\now \in\clockval c$,
  from which the second claim follows.
\end{proof}

%

Since timed automorphisms preserve integer differences,
only the fractional parts of elements of $S \subseteq \Rnonneg$ matter
for $S$-invariance, and hence it makes sense to restrict to subsets of the half-open interval $\ropen{0, 1}$.
Let $\fract{S} = \setof{\fract{x}}{x\in S} \subseteq \ropen{0,1}$ stand for the set of fractional parts of elements of $S$.
The following lemma shows that,
modulo the irrelevant integer parts,
there is always the least set $S$ witnessing $S$-invariance (c.f.~the least support property of, e.g., equality atoms).
\begin{restatable}{lem}{lemLeastSup}%
  \label{lem:leastsup}
  For finite subsets $S, S' \subseteq \Rnonneg$,
  if a timed language $L$ is both $S$-invariant and $S'$-invariant, 
  then it is also $S''$-invariant 
  as long as $\fract{S''} = \fract{S} \cap \fract{S'}$.
\end{restatable}
\begin{proof}
Let $L$ be an $S$- and $S'$-invariant timed language, and let 
%
$F = \fract{S}$ and $F' = \fract{S'}$.
We prove that $L$ is an $(F \cap F')$-invariant subset of $\timedwords \Sigma$.
Consider two timed words $w, w' \in \timedwords \Sigma$ such that $w' = \pi(w)$ for some $(F\cap F')$-timed
automorphism $\pi$.
We need to show
\begin{align*}
  w \in L \quad \text{iff} \quad w' \in L,
\end{align*}
which follows immediately by the following claim:
\begin{clm}\label{claim:toprove}
Every $(F\cap F')$-timed automorphism $\pi$ decomposes into $\pi = \pi_n \circ \dots \circ \pi_1$, where
each $\pi_i$ is either an $F$- or an $F'$-timed automorphism.
\end{clm}
%
%
\noindent
Composition of timed automorphisms makes $\Pi$ into a group.
In short terms, Claim~\ref{claim:toprove} states that $\Pi_{F\cap F'} \subseteq \Pi_F + \Pi_{F'}$, where
$\Pi_F + \Pi_{F'}$ is the smallest subgroup of $\Pi$ including both $\Pi_F$ and $\Pi_{F'}$.
We state below in Claim~\ref{claim:toprovefinal} a fact equivalent to Claim~\ref{claim:toprove},
and which is based on the proof of Theorem 9.3 in~\cite{LMCS14}.
An important ingredient of the proof of Claim~\ref{claim:toprovefinal} is the following fact where,
instead of dealing with decomposition of $\pi$,
we analyse the individual orbit of $F \setminus F'$, in the special case when
both $F \setminus F'$ and $F' \setminus F$ are singleton sets:
%
\begin{clm}\label{claim:proved}
  Let $F, F' \subseteq \ropen{0, 1}$ be finite sets \st $F \setminus F' = \set{t}$ and $F' \setminus F = \set{t'}$.
  For every $(F\cap F')$-timed automorphism $\pi$ we have $\pi(t) = (\pi_n \circ \dots \circ \pi_1)(t)$,
  for some $\pi_1, \dots, \pi_n$, each of which is either an $F$- or an $F'$-timed automorphism
  (i.e., belongs to $\Pi_F + \Pi_{F'}$).
\end{clm}
%
\begin{proof}[Proof of Claim~\ref{claim:proved}]
  We split the proof into two cases.

  \para{Case $F\cap F' \neq \emptyset$}
  Let $l$ be the greatest element of $F\cap F'$ smaller than $t$, and let $h$ be the smallest element of $F\cap F'$
  greater than $t$, assuming they both exist.
  (If $l$ does not exist put $l := h'-1$, where $h'$ is the greatest element of $F\cap F'$;
  symmetrically, if $h$ does not exists put $h := l'+1$, where $l'$ is the smallest element of $F \cap F'$.)
  Then the $(F\cap F')$-orbit $\setof{\pi(t)}{\pi \text{ is an } (F\cap F')\text{-timed automorphism}}$
  is the open interval  $(l,h)$.
  Take any $(F\cap F')$-timed automorphism $\pi$; without loss of generality assume that $u=\pi(t)>t$.
  The only interesting case is $t < t' \leq u$.  
  In this case, we show $\pi(t) = \pi_2(\pi_1(t))$,where
  \begin{itemize}
  \item $\pi_1$ is some $F'$-timed automorphism that acts as the identity on $[t',l+1]$ and \st $t < \pi_1(t) < t'$,
  \item $\pi_2$ is some $F$-timed automorphism that acts as the identity on $[h-1, t]$ and \st $\pi_2(\pi_1(t))=u$.
  \end{itemize}

  \para{Case $F\cap F' = \emptyset$}
  Thus $F = \set{t}$ and $F' = \set{t'}$.
  Take any timed automorphism $\pi$; without loss of generality assume that $\pi(t)>t$.
  Let $z\in\Z$ be the unique integer \st $t' + z -1 < t < t' + z$.
  %
  Let $\pi_1$ be an arbitrary $\set{t'}$-timed automorphism that maps $t$ to some $t_1 \in (t,t'+z)$.
  Note that $t_1$ may be any value in $(t, t'+z)$.
  Similarly, let $\pi_2$ be an arbitrary $\set{t}$-timed automorphism that maps $t_1$ to some $t_2 \in (t', t+1)$.
  Again, $t_2$ may be any value in $(t', t+1)$.
  By repeating this process sufficiently many times one finally reaches $\pi(t)$ as required.
\end{proof}

\begin{clm}\label{claim:toprovefinal}
Let $F, F' \subseteq \ropen{0, 1}$ be finite sets and let $G\subseteq \Pi$ be a subgroup of $\Pi$.
If $\Pi_F\subseteq G$ and $\Pi_{F'} \subseteq G$ then $\Pi_{F\cap F'}\subseteq G$.
\end{clm}
\begin{proof}[Proof of Claim~\ref{claim:toprovefinal}]
  The proof is by induction on the size of the (finite) set $F\cup F'$.
  If $F\subseteq F'$ or $F'\subseteq F$, then the conclusion follows trivially.
  Otherwise, consider any $t \in F \setminus F'$ and $t' \in F' \setminus F$; obviously $t\neq t'$.
  Define $E = (F \cup F') \setminus\set{t,t'}$.
  We have $F \subseteq E\cup\set{t}$ and $F' \subseteq E\cup\set{t'}$ hence
  \[
  \Pi_{E\cup\set{t}} \subseteq \Pi_F \subseteq G \qquad
  \Pi_{E\cup\set{t'}} \subseteq \Pi_{F'} \subseteq G.
  \]
  We shall now prove that $\Pi_E \subseteq G$.
  To this end, consider any $\pi \in \Pi_E$. By Claim~\ref{claim:proved},
  there exists a permutation
  \[
  \tau = \pi_n \circ \dots \circ \pi_1 \in \Pi_{E\cup\set{t}} + \Pi_{E\cup\set{t'}}
  \]
  such that $\pi(t) = \tau(t)$. In other words,
  each of $\pi_1, \dots, \pi_n$ is either a $(E\cup\set{t})$- or a $(E\cup\set{t'})$-timed automorphism.
  Since $\Pi_{E\cup\set{t}} \subseteq G$ and  $\Pi_{E\cup\set{t'}} \subseteq G$, all $\pi_i \in G$, hence also $\tau \in G$.

  On the other hand, clearly $\Pi_{E\cup\set{t}} \subseteq \Pi_E$ and $\Pi_{E\cup\set{t'}} \subseteq \Pi_E$,
  so all $\pi_i \in \Pi_E$, therefore $\tau \in \Pi_E$.
  As a result, $\pi^{-1} \circ \tau \in \Pi_E$.
  Since $(\pi^{-1} \circ \tau)(t) = t$, we obtain $\pi^{-1} \circ \tau \in \Pi_{E\cup\set{t}}$, therefore
  $\pi^{-1} \circ \tau \in G$.
  Together with $\tau \in G$ proved above, this gives $\pi \in G$.
  Thus we have proved $\Pi_E \subseteq G$.

  It is now easy to show that $\Pi_{F\cap F'} \subseteq G$.
  Indeed, $|F \cup E| = |F \cup F'|-1$, so by the
  inductive assumption for $F$ and $E$, we have $\Pi_{F\setminus\set{t}} \subseteq G$
  (note that $F\setminus\set{t} = F\cap E$).
  Further, $|(F\setminus\set{t}) \cup F'| = |F \cup F'|-1$, so $\Pi_{F\cap F'} \subseteq G$
  (note that $(F\setminus\set{t})\cap F' = F\cap F'$), as required.
\end{proof}
Claim~\ref{claim:toprovefinal} immediately implies Claim~\ref{claim:toprove} by taking $G = \Pi_F + \Pi_{F'}$.
Lemma~\ref{lem:leastsup} is thus proved.
 \end{proof}

As a direct corollary of Lemma~\ref{lem:leastsup}, we have:
\begin{cor}%
  \label{cor:leastsup}
  For every timed language $L$, the set
  $\setof{\fract S}{S \subseteq_\text{fin} \Rnonneg, L \text{ is $S$-invariant}}$, if nonempty,
  has a least (inclusion-wise) element.
\end{cor}

Finally, recall $S$-orbits and orbits of elements, as defined abstractly in Section~\ref{sec:atoms}.
Every bounded region corresponds to an orbit of configurations.
Hence,
in case of greedily resetting \NTA where all reachable regions are bounded,
orbits of reachable configurations are in bijective correspondence with reachable regions:
\begin{fact}\label{fact:reg}
  %
  Assume $A$ is a greedily resetting \kmNTA k m.
  Two reachable configurations $(p, \mu, \now)$ and $(p, \mu', \now')$ of $A$ with the same control location $p$
  have the same orbit
  if, and only if,
  the corresponding clock valuations $\assigntoval \mu {\now}$ and
  $\assigntoval {\mu'} {\now'}$
  belong to the same $k,m$-region.
\end{fact}

%
%
%


\section{Decidability of \texorpdfstring{\kDTA{k}}{DTA-k} and \texorpdfstring{\kmDTA k m}{DTA-k,m} membership for \texorpdfstring{\kNTA{1}}{NTA-1}}%
\label{sec:upperbound}

In this section we prove our main decidability result for timed automata, which we now recall.
%
\thmkDTAmemb*
Both results are shown using the following key characterisation of those \kNTA 1 languages which are also \kDTA k languages.
%
%
%
%
\begin{lem}\label{thm:k-DTA-char}
	Let $A$ be a \kmNTA{1} m with $n$ control locations,
	and let $k\in\N$.
	The following conditions are equivalent:
	\begin{enumerate}
		\item $\langts A = \langts B$ for some always resetting \kDTA{k} $B$.\label{p1}
		\item For every timed word $w$, there is $S\subseteq\Rnonneg$ of size at most $k$ \st
		the last timestamp of $w$ is in $S$ and 
		$w^{-1} \langts A$ is $S$-invariant.\label{p2}
		\item $\langts A = \langts B$ for some always resetting \kmDTA k m $B$
			with at most $f(k,m,n) = \reg k m \cdot 2^{n (2km+1)}$ control locations
			($\reg k m$ stands for the number of $k,m$-regions).\label{p3}
	\end{enumerate}
\end{lem}

\noindent
As in case of register automata, this characterisation provides a bound on the number of control locations of a \kDTA k equivalent to a given \kNTA 1 (if any exists).

The proof of \Cref{thm:kDTA:memb} builds on Lemma~\ref{thm:k-DTA-char}
and on the following fact:
\begin{lem}\label{lem:sep}
	The \kDTA k and \kmDTA k m membership problems are both decidable for \DTA languages.
\end{lem}
\begin{proof}
	We reduce to a deterministic separability problem.
	Recall that a language $S$ \emph{separates} two languages $L, M$
	if $L \subseteq S$ and $S \cap M = \emptyset$.
	It has recently been shown that the $\kDTA k$ and $\kmDTA k m$ separability problems
	are decidable for \NTA~\cite[Theorem 1.1]{ClementeLasotaPiorkowski:ICALP:2020},
	and thus, in particular, for \DTA\@.
	To solve the membership problem,
	given a \DTA $A$, the procedure computes a \DTA $A'$ recognising the complement of $\lang A$ and checks whether
	$A$ and $A'$ are \kDTA k separable (resp.,~\kmDTA k m separable)
	by using the result above.
	It is a simple set-theoretic observation that
	$\langts A$ is a \kDTA k language if, and only if, the languages $\langts A$ and $\langts {A'}$ are separated
	by some \kDTA k language, and likewise for \kmDTA k m languages.
\end{proof}
\begin{proof}[Proof of \Cref{thm:kDTA:memb}]
	We solve both problems in essentially the same way.
	Given a \kmNTA 1 m $A$,
	the decision procedure enumerates all always resetting \kmDTA{k+1} m $B$ with at most $f(k+1,m,n)$
	locations
	and checks whether $\langts A = \langts B$
	(which is decidable by~\cite[Theorem 17]{OW04}).
	If no such \kDTA{k+1} $B$ is found,
	the $\langts A$ is not an always resetting \kDTA {k+1} language, due to Lemma~\ref{thm:k-DTA-char},
	and hence forcedly is not a \kDTA k language either;
	the procedure therefore answers negatively.
	Otherwise, in case when such a \kDTA{k+1} $B$ is found,
	then a \kDTA{k} membership
	(resp.~\kmDTA k m membership) test is performed on $B$,
	which is decidable due to Lemma~\ref{lem:sep}.
\end{proof}

\begin{rem}[Complexity]
	The decision procedure for \kNTA{1} invokes the \HyperAckermann~subroutine of~\cite{OW04} to check equivalence between a \kNTA 1 and a candidate \DTA\@.
	This is in a sense unavoidable,
    since we show in \Cref{thm:easy-undecidability} that  \kDTA k and \kmDTA k m membership problems are
    \HyperAckermann-hard for \kNTA{1}.
\end{rem}

In the rest of this section we present the proof of Lemma~\ref{thm:k-DTA-char}.
The proof is a suitable extension and refinement of the argument used in case of register automata in Section~\ref{sec:upperboundreg}.

\subsection{Proof of Lemma~\ref{thm:k-DTA-char}}
Let us fix a \kmNTA{1}{m} $A = \tuple{\Sigma, \L, \set{\x}, \I, \F, \Delta}$, where $m$ is the greatest constant used in clock constraints in $A$, and $k\in\N$.
We assume w.l.o.g.~that $A$ is greedily resetting:
This is achieved by resetting the clock as soon as upon reading an input symbol
its value becomes greater than $m$ or is an integer $\leq m$;
we can record in the control location the actual integral value if it is $\leq m$, or a special flag otherwise.
Consequently, after every discrete transition the value of the clock is at most $m$,
and if it is an integer then it equals 0.

The implication~\ref{p3}$\limplies$\ref{p1} follows by definition.
For the implication~\ref{p1}$\limplies$\ref{p2} suppose, by assumption, 
$\langts A = \langts {B}$ for a total always resetting \kDTA{k} $B$.
Every left quotient $w^{-1} \langts A$ equals $\langtsa B c$ for some configuration $c$, hence
Point~\ref{p2} follows by Fact~\ref{fact:invariantalways}.
Here we use the fact that $B$ is always resetting
in order to apply the second part of Fact~\ref{fact:invariantalways};
without the assumption, we would only have $S$-invariance for sets $S$ of size at most $k+1$.

It thus remains to prove the implication~\ref{p2}$\limplies$\ref{p3},
which will be the content of the rest of the section.
Assuming Point~\ref{p2}, we
are going to define an always resetting \kmDTA k m $B'$ with clocks $\X = \set{\x_1, \dots, \x_{k}}$ and
with at most $f(k,m,n)$ locations such that $\langts {B'} = \langts A$.
We start from the timed transition system $\XX$ obtained by the finite powerset construction underlying
the determinisation of $A$, and then transform
this transition system gradually, while preserving its language,
until it finally becomes isomorphic to the reachable part of $\semd {B'}$ for some \kmDTA k m $B'$.
As the last step we extract from this deterministic timed transition system a syntactic definition of $B'$ and prove equality
of their languages.
This is achieved thanks to the invariance properties of the
transition systems in the course of the transformation.

\para{Macro-configurations}
Configurations of the \kNTA 1 $A$ are of the form $c = (p, u, \now)$ where $u, \now \in \Rnonneg$ and $u \leq \now$.
A \emph{macro-configuration} is a (not necessarily finite)
set $X$ of configurations $(p, u, \now)$ of $A$ which share the same value
of the current timestamp $\now$, which we denote as $\nowof X = \now$.
We use the notation
$\langtsa A X := \bigcup_{c \in X} \langtsa A c$.
%
Let $\succe {\sigma,t} X := \setof{c'\in\Conf A}{c \goesto{\sigma,t} c' \text{ for some }c\in X}$
be the set of successors of configurations in $X$.
We define a deterministic timed transition system $\XX$ consisting of
the macro-configurations reachable in the course of determinisation of $A$.
Let $\XX$ be the smallest set of macro-configurations and transitions such that

\begin{itemize}
\item $\XX$ contains the initial macro-configuration: $X_0 = \setof{(p, 0, 0)}{p\in \I} \in \XX$;
\item $\XX$ is closed under successor: for every $X\in\XX$ and $(\sigma,t)\in\Sigma\times\Rnonneg$, there is a transition
$X \goesto {\sigma,t} \succe {a,t} X$ in $\XX$.
\end{itemize}

\noindent
Due to the fact that $\semd A$ is finitely branching, i.e.~$\succe {\sigma,t} {\set{c}}$ is finite for every fixed $(\sigma,t)$,
all macro-configurations $X\in\XX$ are finite.
Let the final configurations of $\XX$ be $F_\XX = \setof{X\in\XX}{X\cap F \neq \emptyset}$.
\begin{clm}\label{claim:eqlangAX}
$\langtsa A X = \langtsa \XX {X}$ for every $X\in\XX$. In particular $\langts A = \langtsa \XX {X_0}$.
\end{clm}
For a macro-configuration $X$ 
we write
$\macval {X} := \setof{u}{(p, u, \nowof X) \in X}\cup \set{\nowof X}$ to denote the reals appearing in $X$.
Since $A$ is greedily resetting, every macro-configuration
$X \in \XX$ satisfies
$\macval X \subseteq \lopen{\nowof X - m, \nowof X}$.
Whenever a macro-configuration $X$ satisfies this condition
we say that \emph{the span of $X$ is bounded by $m$}.

\para{Pre-states}
%
%
By assumption (Point 2),
$\langtsa A X$ is $S$-invariant for some $S$ of size at most $k$,
but the macro-configuration $X$ itself needs not be $S$-invariant in general.
Indeed, a finite macro-configuration $X\in\XX$ is $S$-invariant if, and only if,
$\fract{\macval X} \subseteq \fract{S}$,
which is impossible in general when $X$ is arbitrarily large, its span is bounded (by $m$), and
size of $S$ is bounded (by $k$).
Intuitively, in order to assure $S$-invariance we will replace $X$ by its $S$-closure $\closure S X$
(recall Fact~\ref{fact:inv:clos}).

A set $S\subseteq \Rnonneg$ is \emph{fraction-independent}
if it contains no two reals with the same fractional part.
A \emph{pre-state} is a pair $Y = (X, S)$,
where $X$ 
is an $S$-invariant macro-state,
and $S$ is a finite fraction-independent subset of $\macval X$ that contains $\nowof X$.
The intuitive rationale behind assuming the $S$-invariance of $X$ is that it implies, together with the bounded span of $X$
and the bounded size of $S$,
that there are only finitely many pre-states,
up to timed automorphism.
%
We define the deterministic timed transition system $\YY$ as the smallest set of pre-states
and transitions between them such that:

\begin{itemize}
\item $\YY$ contains the initial pre-state: $Y_0 = (X_0, \set{0}) \in \YY$; 
\item $\YY$ is closed under the closure of successor: for every $(X,S)\in\YY$ and $(\sigma,t)\in\Sigma\times\Rnonneg$,
there is a transition $(X, S) \goesto {\sigma,t} (X', S')$, where
$S'$ is the least, with respect to set inclusion, subset of $S\cup\set{t}$ containing $t$ such that
the language $L' = (\sigma,t)^{-1} \langtsa A X = \langtsa A {\succe {\sigma,t} X}$ is $S'$-invariant, and
$X' = \closure {S'} {\succe {\sigma,t} X}$.
\end{itemize}


\noindent
Observe that the least such fraction-independent subset $S'$ exists due to the following facts:
by fraction-independence of $S$
there is a unique fraction-independent subset $\widetilde S \subseteq S\cup\set{t}$ which
satisfies $\fract {\widetilde S} = \fract{S\cup\set{t}}$
($\widetilde S$ is obtained by removing from $S$ any element $u$ such that $\fract u = \fract t$, if any);
since $X$ is $S$-invariant, due to Fact~\ref{fact:equivariant:lang} so it is its language
$\langtsa A X$, and hence $L'$ is necessarily $\widetilde S$-invariant;
by assumption (Point 2), $L'$ is $R$-invariant for some set $R\subseteq\Rnonneg$ of size at most $k$ containing $t$;
let $T\subseteq \ropen{0, 1}$ be the least set of fractional values given by Corollary~\ref{cor:leastsup} applied to $L'$, i.e.,
$T \subseteq \fract{\widetilde S} \cap \fract{R}$;
finally let $S'\subseteq \widetilde S$ be chosen so that
$\fract{S'} = T\cup \fract{\set t}$.
Due to fraction-independence of $\widetilde S$ the choice is unique and $S'$ is fraction-independent.
Furthermore, $t\in S'$ and the size of $S'$ is at most $k$.

\begin{exa}
\label{ex:B}
   Suppose $k=3$, $m=2$, $\succe {\sigma,t} X = \set{(p,3.7,5),(q,3.9,5),(r,4.2,5)}$ and
   $S' = \set{3.7, 4.2, 5}$.
   Then $X' = \set{(p,3.7,5)}\cup\set{(q,t,5)\mid t \in (3.7, 4.2)}\cup\set{(r,4.2,5)}$.
   $\nowof {X'} = 5$.
   %
    A corresponding \emph{state} (as defined below)
    is $(X', \mu')$, where $\mu' = \set{\x_1 \mapsto 3.7, \x_2 \mapsto 4.2, \x_3 \mapsto 5}$.
%
\end{exa}

By Fact~\ref{fact:equivariant:lang}, we deduce:
\begin{restatable}[Invariance of $\YY$]{clm}{claimEquivYsecond}%
\label{claim:equivY}
For every two transitions
\[
    (X_1, S_1) \goesto {\sigma,t_1} (X'_1, S'_1)
    \qquad\text{and}\qquad
    (X_2, S_2) \goesto {\sigma,t_2} (X'_2, S'_2)
\]
in $\YY$
and a timed permutation $\pi$, if $\pi(X_1) = X_2$ and $\pi(S_1) = S_2$ and $\pi(t_1) = t_2$, then we have
$\pi(X'_1) = X'_2$ and $\pi(S'_1) = S'_2$.
\end{restatable}
%
\begin{proof}
Let $i$ range over $\set{1, 2}$ and let  $\widetilde X_i := \succe {a,t_i} {X_i}$. Thus $S'_i$ is  the least subset of
$S_i\cup\set{t_i}$ containing $t_i$ such that $\langtsa A {\widetilde X_i}$ is $S'_i$-invariant,
and $X'_i = \closure {S'_i} {\widetilde X_i}$.
By invariance of $\semd A$ (Fact~\ref{fact:equivariant:trans}) and invariance of semantics
(Fact~\ref{fact:equivariant:lang}) we get
\[
\pi(\widetilde X_1) = \widetilde X_2,
\qquad \text{ and } \qquad \pi(\langtsa A {\widetilde X_1}) = \langtsa A {\widetilde X_2},
\]
and therefore $\pi(S'_1) = S'_2$, which implies $\pi(X'_1) = X'_2$.
\end{proof}

Let the final configurations of $\YY$ be $F_\YY = \setof{(X,S)\in\YY}{X\cap \F \neq \emptyset}$.
By induction on the length of timed words it is easy to show:
\begin{clm}\label{claim:eqlangXY}
	$\langtsa \XX {X_0} = \langtsa \YY {Y_0}$.
\end{clm}
Due to the assumption that $A$ is greedily resetting and due to Point 2,
in every pre-state $(X, S) \in \YY$ the span of $X$
is bounded by $m$ and the size of $S$ is bounded by $k$.

\para{States}
We now introduce \emph{states}, which are designed to be in one-to-one correspondence
with configurations of the forthcoming \kDTA k $B'$.
Intuitively, a state differs from a pre-state $(X, S)$ only by allocating the values from $S$ into $k$ clocks,
thus while a pre-state contains a set $S$, the corresponding state contains a reset-point assignment $\mu : \X \to \Rnonneg$
with image $\mu(\X) = S$.

Let $\X = \set{\x_1, \dots, \x_k}$ be a set of $k$ clocks.
A \emph{state} is a pair $Z = (X, \mu)$,
where $X$ 
is a macro-configuration,
$\mu : \X \to \macval X$ is a reset-point assignment,
$\mu(\X)$ is a fraction-independent set containing $\nowof X$,
and $X$ is $\mu(\X)$-invariant.
Thus every state $Z = (X, \mu)$ determines uniquely a corresponding pre-state
$\rho(Z) = (X, S)$ with $S=\mu(\X)$.
We define the deterministic timed transition system $\ZZ$
consisting of those states $Z$ \st $\rho(Z) \in \YY$, and of transitions determined as follows:
$(X, \mu) \goesto {\sigma,t} (X', \mu')$ if the corresponding pre-state has a transition
$(X, S) \goesto {\sigma,t} (X', S')$ in $\YY$, where $S = \mu(\X)$, and
\begin{align}\label{eq:mu}
\mu'(\x_i) \ := \ \begin{cases}
t &\text{ if } \mu(\x_i) \notin S' \text{ or } \mu(\x_i) = \mu(\x_j) \text{ for some } j>i\\
\mu(\x_i) & \text{otherwise.}
\end{cases}
\end{align}
Intuitively, the equation~\eqref{eq:mu} defines a deterministic update of the reset-point assignment $\mu$
that amounts to resetting ($\mu'(\x_i):=t$) all clocks
$\x_i$ whose value is either no longer needed (because $\mu(\x_i) \notin S'$), or
is shared with some other clock $x_j$, for $j>i$ and is thus redundant.
Due to this disciplined elimination of redundancy,
knowing that $t\in S'$ and the size of $S'$ is at most $k$,
we ensure that at least one clock is reset in every step. In consequence, $\mu'(\X) = S'$,
and the forthcoming \kDTA k $B'$ will be always resetting.
Using Claim~\ref{claim:equivY} 
we derive:
\begin{restatable}[Invariance of $\ZZ$]{clm}{claimEquivZsecond}%
\label{claim:equivZ}
For every two transitions
\[
    (X_1, \mu_1) \goesto {\sigma,t_1} (X'_1, \mu'_1)
    \qquad\text{and}\qquad
    (X_2, \mu_2) \goesto {\sigma,t_2} (X'_2, \mu'_2)
\]
in $\ZZ$
and a timed permutation $\pi$, if $\pi(X_1) = X_2$ and $\pi{\circ}{\mu_1} = \mu_2$ and $\pi(t_1) = t_2$, then we have
$\pi(X'_1) = X'_2$ and $\pi{\circ}{\mu'_1} = \mu'_2$.
\end{restatable}
%
\begin{proof}
Let $i$ range over $\set{1, 2}$. Let
$S_i = \mu_i(\X)$ and $(X_i, S_i) \goesto {a,t_i} (X'_i, S'_i)$ in $\YY$. By Claim~\ref{claim:equivY} we have
\[
\pi(X'_1) = X'_2 \qquad \text{ and } \pi(S'_1) = S'_2.
\]
Since $\pi{\circ}{\mu_1} = \mu_2$ and the definition~\eqref{eq:mu} is invariant:
\[
\pi{\circ}(\mu') = (\pi{\circ}{\mu})',
\]
we derive $\pi{\circ}{\mu'_1} = \mu'_2$.
\end{proof}

Let the initial state be $Z_0 = (X_0, \mu_0)$, where $\mu_0(\x_i) = 0$ for all $\x_i \in \X$, and
let the final states be $F_\ZZ = \setof{(X, \mu)\in \ZZ}{X \cap F \neq \emptyset}$.
By induction on the length of timed words one proves:
\begin{clm}\label{claim:eqlangYZ}
$\langtsa \YY {Y_0} = \langtsa \ZZ {Z_0}$.
\end{clm}
%
%
In the sequel we restrict $\ZZ$ to states reachable from $Z_0$.
In every state $Z  = (X, \mu)$ in $\ZZ$,
we have $\nowof X \in \mu(\X)$.
This will ensure the resulting \kDTA k $B'$ to be always resetting.

\para{Orbits of states}

While a state is designed to correspond to a configuration of the forthcoming \kDTA k $B'$,
its orbit is designed to play the r\^ole of control location of $B'$.
We therefore need to prove that the set of states in $\ZZ$ is orbit-finite, i.e.,
the set of orbits
$\setof{\orbit {} Z}{Z\in\ZZ}$
is finite and its size is bounded by $f(k,m,n)$.
We start by deducing an analogue of Fact~\ref{fact:reg}:
\begin{clm}\label{claim:reg}
For two states $Z = (X, \mu)$ and $Z' = (X', \mu')$ in $\ZZ$, their reset-point assignments are in the same orbit,
i.e., $\pi{\circ}\mu = \mu'$ for some $\pi\in\Pi$, if, and only if,
the corresponding  clock valuations $\assigntoval \mu {\nowof X}$ and
$\assigntoval {\mu'} {\nowof {X'}}$
belong to the same $k,m$-region.
\end{clm}
(In passing note that, since in every state $(X, \mu)$ in $\ZZ$ the span of $X$ is bounded by $m$,
only bounded $k,m$-regions can appear in the last claim.
Moreover, in each $k,m$-region one of the clocks constantly equals $0$.)
The action of timed automorphisms on macro-configurations and reset-point assignments is extended to states as $\pi(X, \mu) = (\pi(X), \pi{\circ}\mu)$.
Recall that the orbit of a state $Z$ is defined as $\orbit {}  {Z} = \setof{\pi(Z)}{\pi \in \Pi}$.
\begin{clm}
The number of orbits of states in $\ZZ$ is bounded by $f(k,m,n)$.
\end{clm}
\begin{proof}
	We finitely represent a state $Z = (X, \mu)$, relying on the following general fact.
	\begin{fact}
		For every $u \in\Rnonneg$ and $S\subseteq \Rnonneg$,
		the $S$-orbit\footnote{The orbits of states $Z$ should not be confused with 	$S$-orbits of individual reals $u\in\Rnonneg$.}
		$\orbit S u$ 
		is either the singleton $\set{u}$ (when $u\in S$) or an open interval
		with end-points of the form $t + z$ where $t \in S$ and $z\in\Z$ (when $u\notin S$).
	\end{fact}
	\noindent
	We apply the fact above to $S = \mu (\X)$.
	In our case the span of $X$ is bounded by $m$,
	and thus the same holds for $\mu (\X)$.
	Consequently, the integer $z$ in the fact above
	always belongs to $\set{-m, -m{+}1, \dots, m}$.
	In turn, $X$ splits into disjoint $\clockval X$-orbits 
	$\orbit {\clockval X} u$ consisting of open intervals
	separated by endpoints of the form $t + z$
	where $t \in \mu(\X)$ and $z\in\set{-m, -m{+}1, \dots, m}$.

\begin{exa}
	Continuing Example~\ref{ex:B}, the endpoints are $\set{3, 3.2, 3.7, 4, 4.2, 4.7, 5}$, as shown in the illustration:
	\begin{center}
		\includegraphics[scale=0.5]{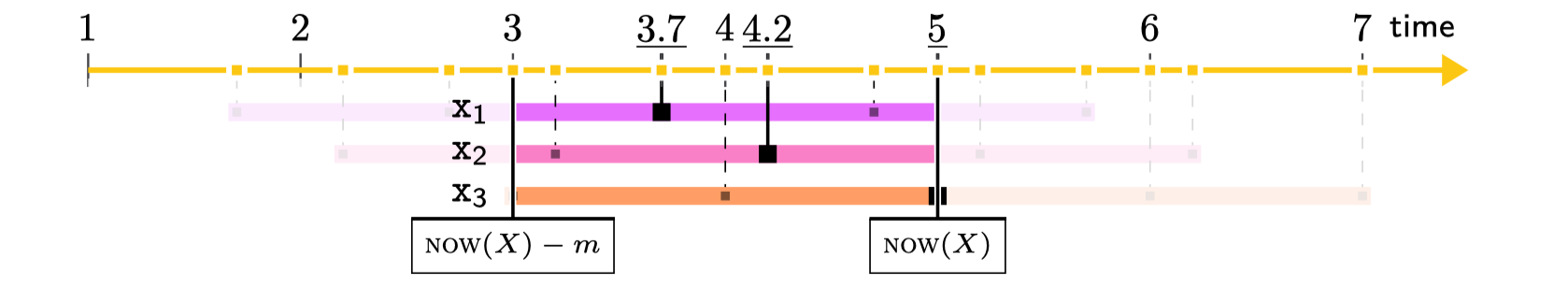}
	\end{center}
\end{exa}

	\noindent
	Recall that $\mu(\X)$ is fraction-independent.
	Let $e_1 < e_2 < \dots < e_{l+1}$ be all the endpoints of open-interval orbits ($l \leq km$), and let 
	$o_1, o_2, o_3, \dots \ := \ \set{e_1}, (e_1, e_2), \set{e_2}, \dots$ be the consecutive $S$-orbits
$\orbit {\mu(\X)} u$ of elements $u\in\clockval \X$.
The number thereof is $2l+1 \leq 2km+1$.
The finite representation of $Z = (X, \mu)$ consists of
the pair $\pair{O}{\mu}$, where
\begin{align}\label{eq:O}
	O = \set{(o_1, P_1), \dots, (o_{2l+1}, P_{2l+1})}
\end{align}
assigns to each orbit $o_i$ the set of locations
$
P_i  = \setof{p}{(p, u, \now)\in X \text{ for some } u\in o_i} \subseteq \L,
$
(which is the same as
$
P_i = \setof{p}{(p, u, \now)\in X \text{ for all } u\in o_i}
$
since $X$ is $\mu(\X)$-invariant, and hence $\mu(\X)$-closed).
Thus a state $Z = (X, \mu)$ is uniquely determined by the sequence $O$ as in~\eqref{eq:O}
and the reset-point assignment $\mu$.

We claim that the set of all the finite representations $(O, \mu)$, as defined above, is orbit-finite.
Indeed, the orbit of $(O, \mu)$ is determined by the orbit of $\mu$ and the sequence
\begin{align}\label{eq:P}
P_1, \ P_2, \ \ldots, \ P_{2km+1}
\end{align}
induced by the assignment $O$ as in~\eqref{eq:O}.
Therefore, the number of orbits is bounded by the number of orbits of $\mu$
(which is bounded, due to Claim~\ref{claim:reg}, by $\reg k m$)
times the number of different sequences of the form~\eqref{eq:P}
(which is bounded by $(2^n)^{2km+1}$).
This yields the required bound $f(k, m, n) = \reg k m \cdot 2^{n(2km +1)}$.
\end{proof}

\para{Construction of the \DTA}  
As the last step we define a \kDTA {k} $B' = \tuple{\Sigma, \L', \X, \set{o_0}, \F', \Delta'}$
such that the reachable part of $\semd {B'}$ is isomorphic to $\ZZ$.
Let locations $\L' = \setof{\orbit {} Z}{Z\in\ZZ}$ be orbits of states from $\ZZ$, the initial location be the orbit $o_0$ of $Z_0$, and final locations
$\F' = \setof{\orbit{}{Z}}{Z\in F_\ZZ}$ be orbits of final states.
A transition $Z = (X, \mu) \goesto{\sigma, t} (X', \mu') = Z'$ in $\ZZ$
induces a transition rule in $B'$
\begin{align}\label{eq:trofB}
	\transition{o}{a}{\psi}{\Y}{o'} \in \Delta'
\end{align}
whenever $o = \orbit{}{Z}$, $o' = \orbit{}{Z'}$,
$\psi$ is the unique $k,m$-region satisfying $\assigntoval \mu t \in \semlog{\psi}$,
and $\Y = \setof{\x_i\in\X}{\mu'(\x_i) = t}$.
The automaton $B'$ is indeed a \DTA since $o$, $\sigma$ and $\psi$ uniquely determine $\Y$ and $o'$:
\begin{clm}
Suppose that two transitions
$(X_1, \mu_1) \goesto {\sigma,t_1} (X'_1, \mu'_1)$ and $(X_2, \mu_2) \goesto {\sigma,t_2} (X'_2, \mu'_2)$ in $\ZZ$
induce transition rules $\transition{o}{\sigma}{\psi}{\Y_1}{o'_1}, \transition{o}{\sigma}{\psi}{\Y_2}{o'_2} \in \Delta'$
with the same source location $o$ and constraint $\psi$, i.e,
\begin{align}\label{eq:psi}
\assigntoval {\mu_1} {t_1} \in\semlog{\psi} \qquad
\assigntoval {\mu_2} {t_2} \in\semlog{\psi}.
\end{align}
Then the target locations are equal $o'_1 = o'_2$, and the same for the reset sets $\Y_1 = \Y_2$.
\end{clm}
\noindent
(Notice that we only consider two transition rules with the same constraint $\psi$, instead of two
different jointly satisfiable constraints $\psi, \psi'$ as in the definition of deterministic timed automata,
due to the fact that  each constraint of $B'$ is a single $k,m$-region.)
\begin{proof}
We use the invariance of semantics of $A$ and Claim~\ref{claim:equivZ}.
Let $o = \orbit{}{X_1, \mu_1} = \orbit{}{X_2, \mu_2}$.
Thus there is a timed automorphism $\pi$ such that
\begin{align}\label{eq:Y}
X_2 = \pi(X_1) \qquad
\mu_2 = \pi{\circ}\mu_1.
\end{align}
%
%
It suffices to show that there is a (possibly different) timed permutation $\pi'$ satisfying the following equalities:
\begin{align}\label{eq:4toprove}
t_2 = \pi'(t_1) \quad
\setof{i}{\mu'_1(\x_i) = t_1} = \setof{i}{\mu'_2(\x_i) = t_2} \quad
\mu'_2 = \pi'{\circ}\mu'_1 \quad
X'_2 = \pi'(X'_1).
\end{align}
We now rely the fact that both ${\now}_1 = \nowof {X_1} \in \mu_1(\X)$ and ${\now}_2 = \nowof {X_2} \in \mu_2(\X)$
are assigned to the same clock due to the second equality in~\eqref{eq:Y}:
${\now}_1 = \mu_1(\x_i)$ and ${\now}_2 = \mu_2(\x_i)$.
We focus on the case when $t_1 - {\now}_1 \leq m$ (the other case is similar and easier since all clocks are reset
due to greedy resetting),
which implies $t_2 - {\now}_2 \leq m$ due to~\eqref{eq:psi}.
In this case we may assume w.l.o.g., due to~\eqref{eq:psi} and the equalities~\eqref{eq:Y}, 
that $\pi$ is chosen so that $\pi(t_1) = t_2$.
We thus take $\pi' = \pi$ for proving the equalities~\eqref{eq:4toprove}.
Being done with the first equality, we observe that the last two equalities in~\eqref{eq:4toprove} hold
due to the invariance of $\ZZ$ (c.f.~Claim~\ref{claim:equivZ}).
The remaining second equality in~\eqref{eq:4toprove} is a consequence of the third one.
\end{proof}
%
%
\begin{clm}%
	\label{claim:last}
	Let $Z = (X, \mu)$ and $Z' = (X', \mu)$ be two states in $\ZZ$ with the same reset-point assignment.
	If $\pi(X) = X'$ and $\pi{\circ}\mu = \mu$ for some timed automorphism $\pi$ then $X = X'$.
\end{clm}
\begin{clm}\label{claim:eqlangCB}
	$\ZZ$ is isomorphic to the reachable part of $\semd {B'}$.
\end{clm}
\begin{proof}
%
	We essentially repeat the argument of Claim~\ref{claim:eqlangCBreg}.
	For a state $Z = (X, \mu)$,
	let $\confof Z = (o, \mu, t)$, where $o = \orbit{}{Z}$ and $t = \nowof X$.
	Let $\ZZ'$ denote the reachable part of $\semd {B'}$.
	By Claim~\ref{claim:lastreg},
	the mapping $\confof \_$ is a bijection between $\ZZ$ and its image $\confof{\ZZ} \subseteq \semd {B'}$.
	We aim at proving $\confof \ZZ = \ZZ'$.

	By the very definition~\eqref{eq:trofBreg}, the image
	$\confof \ZZ$ is a subsystem of $\ZZ'$.
	Recall that $\ZZ$ is total: for every
	$(\sigma_1, t_1) \dots (\sigma_n, t_n) \in\timedwords \Sigma$,
	there is a sequence of transitions $(X_0, \mu_0) \goesto {\sigma_1, t_1} \cdots \goesto{\sigma_n, t_n}$ in $\ZZ$.
	Therefore $\confof \ZZ$ is total too and, since $\ZZ'$ is deterministic and reachable, the subsystem $\confof \ZZ$ necessarily equals $\ZZ'$.
\end{proof}
Claims~\ref{claim:eqlangAX},~\ref{claim:eqlangXY},~\ref{claim:eqlangYZ},
and~\ref{claim:eqlangCB} imply $\langts A = \langts {B'}$, which completes the proof
of Lemma~\ref{thm:k-DTA-char}.
%


\section{Undecidability and hardness}%
\label{sec:lowerbound}

In this section we complete the decidability and complexity landscape for the deterministic membership problem
by providing matching undecidability and complexity hardness results, both for register and timed automata.


\subsection{Lossy counter machines}

Our undecidability and hardness results will be obtained by reducing from the finiteness problem for lossy counter machines, which is known to be undecidable.
A \emph{$k$-counters lossy counter machine} (\kLCM k) is a tuple $M = \tuple {C, Q, q_0, \Delta}$,
where $C = \set{c_1, \dots, c_k}$ is a set of $k$ counters,
$Q$ is a finite set of control locations,
$q_0 \in Q$ is the initial control location,
and $\Delta$ is a finite set of instructions of the form $\tuple {p, \op, q}$,
where $\op$ is one of $\incr c$, $\decr c$, and $\ztest c$.
A configuration of an \LCM $M$ is a pair $\tuple {p, u}$,
where $p \in Q$ is a control location,
and $u \in \N^C$ is a counter valuation.
For two counter valuations $u, v \in \N^C$,
we write $u \leq v$ if $u(c) \leq v(c)$ for every counter $c \in C$.
The semantics of an \LCM $M$ is given by a (potentially infinite) transition system over the configurations of $M$
\st there is a transition $\tuple {p, u} \goesto \delta \tuple {q, v}$,
for
$\delta = \tuple {p, \op, q} \in \Delta$, whenever
%
\begin{enumerate}
    \item $\op = \incr c$ and $v \leq \extend c {u(c) + 1} u$, or
    \item $\op = \decr c$ and $v \leq \extend c {u(c) - 1} u$, or
    \item $\op = \ztest c$ and $u(c) = 0$ and $v \leq u$.
\end{enumerate}
We omit $\delta$ in  $\goesto \delta$ when it is irrelevant, and write $\goesto{}^*$ for the transitive closure of $\goesto{}$.
The \emph{finiteness problem} (a.k.a.~space boundedness) for an \LCM $M$
amounts to deciding whether the reachability set
$\reachset M = \setof{\tuple{p, u}} {\tuple {q_0, u_0} \goesto {}^* \tuple{p, u}}$
is finite, where $u_0$ is the constantly $0$ counter valuation.
\begin{thmC}[\protect{\cite[Theorem 13]{Mayr:TCS:2003}}]%
	\label{thm:Mayr:2003}
    The \kLCM 4 finiteness problem is undecidable.
\end{thmC}

\subsection{Register automata}\label{sec:reg-lower}

\subsubsection{Undecidability of \DRA membership for \kNRA 1}
\label{sec:DRA:NRA:undecidability}

We show that it is undecidable whether a \kNRA 1 language can be recognised by some \DRA\@.
%
In the following, it will be useful to be able to refer to projection of a data word on the finite component of the alphabet.
To this end, for a data word $w = \tuple{\sigma_0, a_0} \cdots \tuple{\sigma_n, a_n} \in (\Sigma \times \A)^*$,
let $\undata w = \sigma_0 \cdots \sigma_n \in \Sigma^*$ be the word obtained by removing the atom component.
As already announced, we reduce from the finiteness problem for lossy counter machines.
Consider a lossy counter machine $M = \tuple {C, Q, q_0, \Delta}$ with $4$ counters $C = \set{c_1, c_2, c_3, c_4}$.
We use the following encoding of \LCM runs
as data words over the alphabet $\Sigma = Q \cup \Delta \cup C$
comprising the control locations, transitions, and counters of $M$.
(A similar encoding has been used in the proof of Theorem 5.2 in~\cite{DL09}
to show that the universality problem for \kNRA 1 is not primitive recursive.)
We encode a counter valuation $u \in \N^C$
as the word over $C \subseteq \Sigma$
\begin{align}%
	\label{eq:valuation:enc}
    \enc u \;=\;
        \underbrace {c_1 c_1 \cdots c_1}_{u(c_1) \text{ letters}}\
            \underbrace {c_2 c_2 \cdots c_2}_{u(c_2) \text{ letters}}\
				\underbrace {c_3 c_3 \cdots c_3}_{u(c_3) \text{ letters}}\
					\underbrace {c_4 c_4 \cdots c_4}_{u(c_4) \text{ letters}} \in \set {c_1}^* \set {c_2}^* \set {c_3}^* \set {c_4}^*.
\end{align}
Consider a \LCM run
\begin{align*}
    \pi \;=\; \tuple{p_0, u_0}
        \goesto {\delta_1} \tuple{p_1, u_1}
            \goesto {\delta_2} \cdots
                \goesto {\delta_n} \tuple{p_n, u_n}.
\end{align*}

The set $\Enc \pi \subseteq (\Sigma \times \A)^*$ of \emph{reversal-encodings} of $\pi$
contains all data words $w \in (\Sigma \times \A)^*$
satisfying the following conditions:

\renewcommand{\labelenumii}{(E\arabic{enumi}.\arabic{enumii})}
\renewcommand{\labelenumiii}{(E\arabic{enumi}.\arabic{enumii}.\arabic{enumiii})}

\begin{enumerate}[label=(E\arabic*)]

	\item\label{cond:1}
	the finite part of $w$ is of the form
	\begin{align}%
		\label{eq:undata}
		\undata w = p_n \delta_n \enc {u_n}\quad \cdots\quad p_1 \delta_1 \enc {u_1}\quad p_0 \enc{u_0} \in \Sigma^*;
	\end{align}

    \item\label{cond:2}
	all atoms appearing in a single $\enc {u_i}$ block are distinct;

    \item\label{cond:3}
    for every transition $\delta_i = \tuple{p_{i-1}, \op_i, p_i}$
    and for every counter $c_j \in \set {c_1, c_2, c_3, c_4}$:

    \begin{enumerate}

        \item\label{cond:3:1}
        if $\op_i = \incr {c_j}$ increments counter $c_j$,
		then (recalling that the lossy semantics amounts to $u_i(c_j) - 1 \leq u_{i-1}(c_j)$)
		we require that for each occurrence of $c_j$ in $\enc {u_i}$
        there is a matching occurrence of $c_j$ in $\enc {u_{i-1}}$ with the same atom,
        with the exception of the \emph{last} occurrence of $c_j$ in $\enc {u_i}$;

        \item\label{cond:3:2}
        if $\op_i = \decr {c_j}$ decrements counter $c_j$,
		then  (recalling that the lossy semantics amounts to $u_i(c_j) + 1 \leq u_{i-1}(c_j)$)
		we require that  for each occurrence of $c_j$ in $\enc{u_i}$ 
            	there is a matching occurrence of $c_j$ in $\enc{u_{i-1}}$ with the same atom, which is \emph{not the last}
		occurrence of $c_j$ in $\enc{u_{i-1}}$.
        %
%
%
%

        \item\label{cond:4:3}
        if $\op_i = \ztest {c_j}$ tests whether counter $c_j$ is zero,
        then $u_{i-1}(c_j) = u_i(c_j) = 0$,
		and thus we require that neither $\enc {u_i}$ nor $\enc {u_{i-1}}$ contain any occurrence of $c_j$;

        \item\label{cond:4:4}
        otherwise the operation $\op_i$ does not modify counter $c_j$,
		i.e., $u_i(c_j) \leq u_{i-1}(c_j)$,
		and we require that for each occurrence of $c_j$ in $\enc {u_i}$
        		there is a matching occurrence of $c_j$ in $\enc {u_{i-1}}$ with the same atom.

    \end{enumerate}

\end{enumerate}

\noindent
The intuition underlying the first two items in condition~\ref{cond:3} is that the effect of an increment or decrement
of $c_j$ is encoded by creating or removing \emph{the last} occurrence of $c_j$ in a block.

Let $L = \bigcup_{\pi \text{ a run of } M} \Enc\pi$ be the set of all reversal-encodings of runs of $M$.
Under this encoding, we can build a \kNRA 1 $A$ recognising the \emph{complement} of $L$.
Indeed, $A$ can determine bad encodings $w \not\in L$ by guessing one of finitely many reasons for this to occur:
\renewcommand{\labelenumii}{(F\arabic{enumi}.\arabic{enumii})}
\renewcommand{\labelenumiii}{(F\arabic{enumi}.\arabic{enumii}.\arabic{enumiii})}

\begin{enumerate}[label=(F\arabic*)]

	\item the projection $\undata {\enc \pi}$ to the finite alphabet $\Sigma$ is not a word in the regular language
	$(Q \Delta \set{c_1}^*\set{c_2}^*\set{c_3}^*\set{c_4}^*)^*\set{q_0}$
	(notice that $\enc {u_0}$ is the empty string, since the initial valuation $u_0$ assigns $0$ to every counter),
	or there is a transition $\delta_i = \tuple{p, \_, q}$ \st either the source is incorrect $p \neq p_{i-1}$
	or the destination is incorrect $q \neq p_i$.
	This is even a regular language of finite words over $\Sigma$ (i.e., without atoms).
	Thus in the remaining cases below we can assume that the finite part of $w$ is of the form as in~\eqref{eq:undata};

	\item there is a block $\enc {u_i}$ containing the same atom twice;

    \item 
    there is a transition $\delta_i = \tuple{p_{i-1}, \op_i, p_i}$
    and a counter $c_j \in \set {c_1, c_2, c_3, c_4}$
	\st one of the following condition holds:

    \begin{enumerate}

        \item 
        $\op = \incr {c_j}$ 
		but there is a non-last occurrence of $c_j$ in $\enc {u_i}$
        		without  matching occurrence of $c_j$ in $\enc {u_{i-1}}$ with the same atom;

        \item 
        $\op = \decr {c_j}$ 
        			but some occurrence of $c_j$ in $\enc{u_i}$ either has no matching
			occurrence of $c_j$ in $\enc{u_{i-1}}$ with the same atom, or has a matching occurrence which
			is the last occurrence of $c_j$ in $\enc{u_{i-1}}$.
%
%
%
%

        \item 
        $\op = \ztest {c_j}$ 
        but either $\enc {u_i}$ or $\enc {u_{i-1}}$ contains an occurrence of $c_j$;

        \item 
        the operation $\op$ does not modify counter $c_j$,
		but there is an occurrence of $c_j$ in $\enc {u_i}$ without a
		matching occurrence of $c_j$ in $\enc {u_{i-1}}$ with the same atom.

    \end{enumerate}

\end{enumerate}
One register is sufficient to recognise each of the possible mistakes above.

\begin{lem}%
	\label{lem:NRAReduction}
	The set of reachable configurations $\reachset M$ is finite
	if, and only if,
	$\lang A$ is a \DRA language.
\end{lem}

\begin{proof}
	The ``only if'' implication follows from the fact that if $\reachset M$ is finite, i.e., there is a finite bound $k$ on the
	sum of values of all counters of $M$ in a run,
	then the complement of $\lang A$ (which encodes all correct reversal-encodings) is a \kDRA k language,
	and thus $\lang A$ itself is a \DRA language since \DRA languages are closed under complement.
	For the ``if'' implication, assume that $\reachset M$ is infinite,
	and by way of contradiction assume that $\lang A$ is a \DRA language.
	In this case, the complement $L$ of $\lang A$ is recognised by some \kDRA k $B$ with a finite number of registers $k$.
	Since $\reachset M$ is infinite, there are runs where counter (say) $c_1$ is unbounded.
	In particular, there is a run $\pi$ reaching a counter valuation $u_i$ with $u_i(c_1) \geq k + 2$.
	When $B$ reads the corresponding encoding $w \in \Enc \pi$,
	after reading the first $k+1$ atoms $a_1, \dots, a_{k+1}$ in the $\enc{u_i}$ block,
	it must forget at least one such atom, say $a_j$.
	The next atom $a_{k+2}$ in the $\enc{u_i}$ block can thus be replaced by $a_j$
	and $B$ still accepts the corresponding data word $w'$.
	However, $w'$ is not the reversal encoding of any run of $M$,
	since it violates condition~\ref{cond:2}.
	This contradicts that $B$ recognises $L$,
	and thus $\lang A$ is not a \DRA language,
	as required.
\end{proof}

We thus have a reduction from the \LCM finiteness problem to the \DRA membership problem for \kNRA 1,
which is thus undecidable thanks to \cref{thm:Mayr:2003}.

\begin{cor}
    The \DRA membership problem for \kNRA 1 languages is undecidable.
\end{cor}

\subsubsection{Undecidability and hardness for \kDRA k membership}%
\label{sec:DRA:hardness}

The \emph{\NRA universality problem} asks whether a given \NRA $A$ recognises every data word $\lang A = (\Sigma \times \A)^*$.
%
All the lower bounds in this section leading to undecidability and hardness results for the \kDRA k membership problem
are obtained by a reduction from the universality problem for corresponding classes of data languages.
%
%

\begin{restatable}[\protect{c.f.~\cite[Theorem 1]{Finkel:FORMATS:2006}}]{lem}{lemEasyNRAUndec}\label{lem:easy-NRA-undecidability}
	Let $k\in \N$ and let $\mathcal Y$ be a class of invariant data languages that
	\begin{enumerate}
		\item contains all the \kDRA 0 languages,
		\item is closed under union and concatenation, and 
		\item contains some non-\kDRA k language.
	\end{enumerate}
	The universality problem for data languages in $\mathcal Y$ reduces in polynomial time
	to the \kDRA k membership problem for data languages in $\mathcal Y$.
\end{restatable}

\begin{proof}
	Let $L \in \mathcal Y$ be a data language over a finite alphabet $\Sigma$.
	We show that universality of $L$ reduces to \kDRA k membership.
	Thanks to the last assumption,
	let $M \in \mathcal Y$ be a data language over some finite alphabet $\Gamma$ which is not recognised by any \kDRA k.
	Consider the following language over the extended alphabet $\Sigma' = \Sigma \cup \Gamma \cup \{\$\}$:
	\begin{align*}
		N \ := \ L \cdot (\set{\$} \times \A) \cdot (\Gamma \times \A)^* \, \cup \,  (\Gamma \times \A)^* \cdot (\set{\$} \times \A) \cdot M,
	\end{align*}
   where $\$ \not\in \Sigma \cup \Gamma$ is a fixed fresh alphabet symbol.
   Since $\mathcal Y$ contains the universal language,
   by its closure properties the language $N$ belongs to $\mathcal Y$.
   We conclude by proving the following equivalence:
   \begin{align*}
	   L=(\Sigma \times \A)^* \quad \text{if, and only if, } \quad \text{$N$ is recognised by a \kDRA k.}
	\end{align*}
	For the ``only if'' direction, if $L$ is universal,
	then $N = (\Sigma \times \A)^* \cdot (\set{\$} \times \A) \cdot (\Sigma \times \A)^*$
	is clearly recognised by a \kDRA k.
	For the ``if'' direction suppose, towards reaching a contradiction,
	that $N$ is recognised by a \kDRA k $A$
	but $L$ is not universal.
	Choose an arbitrary data word $w \not\in L$ over $\Sigma$
	and consider an arbitrary extension $u = w \cdot (\$, a)$ of $w$ by one letter.
	Since $\$$ does not belong to the finite alphabet $\Sigma \cup \Gamma$,
	the left quotient $u^{-1} N = \setof v {u v \in N}$ equals $M$.
	Let $(p, \mu)$ be the configuration reached by $A$ after reading $u$,
	which thus recognises $\lang{p, \mu} = M$.
	%
	Since $M$ is invariant as a language in $\mathcal Y$,
	$M$ is a \kDRA k language, which is a contradiction.
\end{proof}

From Lemma~\ref{lem:easy-NRA-undecidability} we immediately obtain the
undecidability and hardness results for the \kDRA k membership problem,
which we now recall.

\thmDRAlowerbounds*

\begin{proof}
	For the first point, consider the class $\mathcal Y$ consisting of all the \kNRA 2 languages. 
	Clearly this class contains all $\kDRA 0$ languages and it is closed under union and concatenation.
	Thanks to Example~\ref{example:L1reg} we know that there are $\kNRA 1$ (and thus $\kNRA 2$) languages which are not \DRA languages.
	Thus the conditions of Lemma~\ref{lem:easy-NRA-undecidability} are satisfied and the universality problem for \kNRA 2 reduces in polynomial time to the \kDRA k membership problem for $\kNRA 2$.
	Since the former problem is undecidable~\cite[Theorem 5.4]{DL09}, undecidability of the latter one follows.
	For the other two points we can proceed in an analogous way,
	by using the fact that the universality problem is undecidable for \kNRAg 1 (\kNRA 1 with guessing)~\cite[Exercise 9]{atombook},
	and not primitive recursive for \kNRA 1~\cite[Theorem 5.2]{DL09}.
\end{proof}

\subsection{Timed automata}

We now focus on timed automata.
Following similar lines as in case of register automata in Section~\ref{sec:reg-lower},
in \cref{sec:undecidability} we prove undecidability of the \DTA membership problem for \kNTA 1 (c.f.~\cref{thm:undecidability})
and in \cref{sec:hardness} we prove \HyperAckermann-hardness of the \kDTA{k} membership problem for \kNTA{1} (c.f.~\cref{thm:easy-undecidability}).

\subsubsection{Undecidability of \DTA and \mDTA m membership for \kNTA{1}}
\label{sec:undecidability}

It has been shown in~\cite[Theorem 1]{Finkel:FORMATS:2006} that it is undecidable
whether a \kNTA 2 timed language can be recognised by some \DTA\@. 
This was obtained by a reduction from the \kNTA 2 universality problem,
which is undecidable.  
While the universality problem becomes decidable for $k = 1$,
we show in this section that, as announced in \cref{thm:undecidability},
the \DTA membership problem remains undecidable for \kNTA 1.



Since the universality problem for \kNTA 1 is decidable,
we need to reduce from another (undecidable) problem.
As in the case of register automata,
we reduce from the finiteness problem of a \LCM $M$ with 4 counters,
which is undecidable by \cref{thm:Mayr:2003}.
In the case of timed automata,
we can use the reversal encoding from~\cite[Definition 4.6]{LasotaWalukiewicz:ATA:ACM08}
showing that we can build a \kNTA 1 $A$ recognising the \emph{complement} of the set of reversal-encodings of the runs of $M$.
Since this construction has already been presented in full details in~\cite[App.~C]{ClementeLasotaPiorkowski:arXiv:2020}
and since it is in complete analogy to the construction in \cref{sec:DRA:NRA:undecidability} for register automata,
we omit it here.
One can then prove the following property,
which is analogous to Lemma~\ref{lem:NRAReduction} in the case of register automata.
\begin{lem}%
	\label{lem:LCM:DTA:reduction}
	The set of reachable configurations $\reachset M$ is finite
	if, and only if,
	$\lang A$ is a deterministic timed language.
\end{lem}



Since the timed automaton constructed in the reduction above uses only constant 1,
the reduction works also for the
\mDTA m membership problem for every fixed $m> 0$,
thus proving \cref{thm:undecidability}.
This result is the best possible in terms of the parameter $m$,
since the problem becomes decidable for $m=0$.
In fact, the class of \kmDTA k 0 languages coincides with the class of \kmDTA 1 0 languages
(one clock is sufficient; c.f.~\cite[Lemma 19]{OW04}),
and thus \mDTA 0 membership reduces to \kmDTA 1 0 membership,
which is decidable for \kNTA 1 by \cref{thm:kDTA:memb}.


\subsubsection{Undecidability and hardness for \kDTA k and \kmDTA k m membership}%
\label{sec:hardness}

All the lower bounds in this section
are obtained by a reduction from the universality problem for suitable language
classes
(does a given timed language $L \subseteq \timedwords{\Sigma}$ satisfy $L = \timedwords{\Sigma}$?),
in complete analogy to \cref{sec:DRA:hardness} dealing with register automata.
Our starting point is the following result.
\begin{thmC}[\protect{\cite[Theorem 1]{Finkel:FORMATS:2006}}]\label{thm:Finkel}
	The \DTA membership problem is undecidable for \NTA languages.
\end{thmC}

We provide a suitable adaptation, generalization, and simplification of the result above
which will allow us to extend undecidability to the $\kDTA k$ membership problem for every fixed $k \geq 0$,
and also to obtain a complexity lower bound for $\kNTA 1$ input languages.
Fix two languages $L \subseteq \timedwords \Sigma$ and $M \subseteq \timedwords \Gamma$,
and a fresh alphabet symbol $\$ \not\in \Sigma \cup \Gamma$.
The \emph{composition} $L \rhd M$
is the timed language over $\Sigma' = \Sigma \cup\set{\$} \cup\Gamma$
defined as follows:
\begin{align*}
	L \rhd M \ = \ \setof{u (\$, t) (v + t)\in \timedwords{\Sigma'}}{u \in L, v \in M, t \in \Rnonnegpos},
\end{align*}
where $t$ is necessarily larger or equal than the last timestamp of $u$ by the definition of $\timedwords{\Sigma'}$.
%
The following lemma exposes some abstract conditions on classes of timed languages
which are sufficient to encode the universality problem.

\begin{restatable}{lem}{lemEasyUndec}\label{lem:easy-undecidability}
	Let $k ,m \in \N$ and let $\mathcal Y$ be a class of timed languages that
	\vspace{-\topsep} \begin{enumerate}
	\item contains all the \kDTA 0 languages,
	\item is closed under union and composition, and 
	\item contains some non-\kDTA{k} (resp.~non-\kmDTA k m) language.
	\end{enumerate}
	The universality problem for languages in $\mathcal Y$ reduces in polynomial time
	to the \kDTA{k} (resp. \kmDTA k m) membership problem for languages in $\mathcal Y$.
\end{restatable}

Lemma~\ref{lem:easy-undecidability} is entirely analogous to Lemma~\ref{lem:easy-NRA-undecidability} for data languages,
except that invariance of languages in $\mathcal Y$ is not required; moreover,
notice that the notion of composition of timed languages that we need to state and prove the lemma above
is a bit more complicated than the straightforward notion of concatenation
that appears in the analogous statement for data languages from Lemma~\ref{lem:easy-NRA-undecidability}.

\begin{proof}
	We consider \kDTA k membership (the \kmDTA k m membership is treated similarly).
	Consider some fixed timed language $M \in \mathcal Y$ which is not recognised by any \kDTA{k}
   (relying on the assumption 3), over an alphabet $\Gamma$.
%
%
   For a given timed language $L \in {\mathcal Y}$, over an alphabet $\Sigma$,
   we construct
   the following language over the extended alphabet $\Sigma \cup \Gamma \cup \{\$\}$:
   %
   \begin{align*}
	   N \; := \; L \rhd \timedwords{\Gamma} \;\cup\; \timedwords{\Sigma} \rhd M \ \subseteq \ \timedwords{\Sigma \cup \Gamma \cup \{\$\}},
   \end{align*}
   where $\$ \not\in \Sigma \cup \Gamma$ is a fixed fresh alphabet symbol.
   Since $\mathcal Y$ contains all the \kDTA 0 languages thanks to the assumption 1,
   and it is closed under union and composition thanks to the assumption 2,
   the language $N$ belongs to $\mathcal Y$.
   We claim that the universality problem for $L$ is equivalent to the \kDTA k membership problem for $N$:
   \begin{align*}
	   L=\timedwords{\Sigma} \quad \textrm{if, and only if,} \quad N \textrm{ is recognised by a \kDTA k}.
   \end{align*}
   For the ``only if'' direction, if $L = \timedwords{\Sigma}$
   then clearly $N = \timedwords{\Sigma} \rhd \timedwords{\Gamma}$.
   Thus $N$ is a \kDTA 0 languages, and thus also \kDTA k for any $k\geq 0$.
   For the ``if'' direction suppose, towards reaching a contradiction,
   that $N$ is recognised by a \kDTA{k} $A$
   but $L\neq\timedwords{\Sigma}$.
   Assume, w.l.o.g., that $A$ is greedily resetting.
   Choose an arbitrary timed word $w = (\sigma_1, t_1) \dots (\sigma_n, t_n) \not\in L$ over $\Sigma$.
   Therefore, for any extension $v = (\sigma_1, t_1) \dots (\sigma_n, t_n) (\$, t_n + t)$ of $w$ by one letter,
   we have
   \[  v^{-1} N = t + M = \setof{(\sigma_1', t+u_1)\dots(\sigma_m', t+u_m)}{(\sigma_1', u_1)\dots(\sigma_m', u_m) \in M}. \] 
   Choose $t$ larger than the largest absolute value $m$ of constants appearing in clock constraints in $A$, and
   let $(p, \mu)$ be the configuration reached by $A$ after reading $v$.
   Since $t > m$, all the clocks are reset by the last transition and hence $\mu(\x) = 0$ for all clocks $\x$.
   Consequently, if the initial control location of $A$ were moved to the location $p$, the so modified \kDTA{k} $A'$ would accept
   the language $M$.
   But this contradicts our initial assumption that $M$ is not recognised by a \kDTA{k},
   thus finishing the proof.
\end{proof}

We can now prove the following refinement of \cref{thm:Finkel} claimed in the introduction.
 \thmEasyUndecidability*
%
\begin{proof}
	%
	Each of the three points follows by an application of Lemma~\ref{lem:easy-undecidability}.
	For instance, for the first point 
	take as $\mathcal Y$ the class of languages recognised by \kNTA 2.
	This class contains all \kDTA 0 languages,
	is closed under union and composition, and
	is not included in \kDTA k for any $k$ nor in \kmDTA k m for any $k,m$
	(c.f.~the \kNTA 1 language from Example~\ref{example:L1} which is not recognised by any \DTA).
	Since the universality problem is undecidable for \kNTA 2~\cite[Theorem 5.2]{AD94},
	by Lemma~\ref{lem:easy-undecidability} the \kDTA k and \kmDTA k m membership problems are undecidable for \kNTA 2.
	The second and third points follow in the same way,
	using the fact that universality is undecidable for \kNTAe 1 (\kNTA 1 with epsilon transitions)~\cite[Theorem 5.3]{LasotaWalukiewicz:ATA:ACM08},
	resp., \HyperAckermann-hard for \kNTA 1
	(by combining the same lower bound for the reachability problem in lossy channel systems~\cite[Theorem 5.5]{CS08},
	together with the reduction from this problem to universality of \kNTA{1} given in~\cite[Theorem 4.1]{LasotaWalukiewicz:ATA:ACM08}).
\end{proof}


\section{Conclusions}

We have shown decidability and undecidability results
for several variants of the deterministic membership problem for timed and register automata.
Regarding undecidability, we have extended the previously known results~\cite{Finkel:FORMATS:2006,Tripakis:IPL:2006}
by proving that the \DTA membership problem is undecidable already for \kNTA 1 (\cref{thm:undecidability}).
%
%
Regarding decidability, we have shown that when the resources available to the deterministic automaton are fixed
(either just the number of clocks $k$, or both clocks $k$ and maximal constant $m$),
then the respective deterministic membership problem is decidable (\cref{thm:kDTA:memb})
and \HyperAckermann-hard (\cref{thm:easy-undecidability}).
We have depicted a similar scenario for register automata,
in regards of both decidability (\cref{thm:kDRA:memb}),
undecidability and hardness (\cref{thm:undecidabilityreg,thm:reg-lower-bound}).

Our deterministic membership algorithm is based on a characterisation of \kNTA 1 languages
which happen to be \kDTA k (\cref{thm:k-DTA-char}),
which is proved using a semantic approach leveraging on notions from the theory of sets with atoms~\cite{BL12}.
It would be interesting to compare this approach
to the syntactic determinisation method of~\cite{BBBB:ICALP:2009}.
Analogous decidability results for register automata have been obtained with similar techniques.

Finally, our decidability results extend to the slightly more expressive class of always resetting \kNTA 2,
which have intermediate expressive power strictly between \kNTA 1 and \kNTA 2.

\section*{Acknowledgment}
We thank S.~Krishna for fruitful discussions. 

\bibliographystyle{alphaurl}

\bibliography{bib}


\end{document}